\documentclass[11pt]{scrartcl}
\usepackage[english]{babel}
\usepackage[T1]{fontenc}
\usepackage[utf8]{inputenc}
\usepackage{float}
\usepackage{color,amssymb,amsfonts,amsthm,mathrsfs,graphicx,setspace,amsmath}
\usepackage{fancyhdr}
\usepackage{caption}
\usepackage{subcaption}
\usepackage[left=3cm,right=3cm,top=2.5cm,bottom=3cm]{geometry}
\usepackage{mathtools}
\usepackage{bm}
\usepackage{xcolor}
\usepackage{tikz}
\usepackage{newunicodechar}

\title{Directed Ear Anonymity}
\author{Marcelo Garlet Milani}
\date{\today}

\usepackage[pdfdisplaydoctitle,menucolor=orange!40!black,filecolor=magenta!40!black,urlcolor=green,linkcolor=red!40!black,citecolor=green!40!black,colorlinks]{hyperref}
\makeatletter
\hypersetup{pdftitle={\@title}, pdfauthor={\@author}}
\makeatother

\usepackage{nameref}
\usepackage[capitalise]{cleveref}

\usepackage[square, numbers, comma, sort&compress]{natbib}  \usepackage{enumitem}
\usepackage{dsfont}
\usepackage{soul}
\usepackage{xspace}
\usepackage[prependcaption,colorinlistoftodos]{todonotes}
\usepackage{algorithm}
\usepackage[noend]{algpseudocode}
\usepackage{environ}

\usepackage{colorpalette,appendix-tools,extra-cref,extra-enumitem,extra-environ,extra-theorems,pictikzgraph,tasks,graphs,algorithms,computational-problems,basic-math,unicode}
\appendixfalse

\hypersetup{
	colorlinks=true,
	linkcolor=darkblue,
	citecolor=darkblue,
	urlcolor=myLightBlue,
	bookmarksopen=true,
	bookmarksnumbered,
	bookmarksopenlevel=2,
	bookmarksdepth=3
}

\usetikzlibrary[positioning]
\usetikzlibrary{patterns}
\usetikzlibrary{shapes}
\usetikzlibrary{decorations.markings}
\usetikzlibrary{arrows}
\usetikzlibrary{calc}

\begin{document}
\maketitle

\begin{abstract}
	We define and study a new structural parameter for directed graphs, which we call
	\emph{ear anonymity}.
	Our parameter aims to generalize the useful properties of \emph{funnels} to larger digraph classes.
	In particular, funnels are exactly the acyclic digraphs with ear anonymity one.
	We prove that computing the ear anonymity of a digraph is \NP/-hard and
	that it can be solved in $\Bo(m(n + m))$-time on acyclic digraphs (where \(n\) is the number of vertices and \(m\) is the number of arcs in the input digraph).
	It remains open where exactly in the polynomial hierarchy the problem of computing ear anonymity lies,
	however for a related problem we manage to show $\Sigma_2^p$-completeness.
\end{abstract}

\section{Introduction}
One approach for handling computationally hard problems is to design algorithms which
are efficient if certain structural parameters of the input are small.
In undirected graphs, \emph{width} parameters such as \emph{treewidth} \cite{RobertsonS86,Bodlaender07,courcelle1990monadic} and \emph{cliquewidth} \cite{CourcelleO00}
are very effective in handling a number of problems (see also \cite{downey2013fundamentals}).

Width parameters for directed graphs, however, seem to be less powerful \cite{GanianHK0ORS16,GanianHKLOR14}.
While \emph{directed treewidth} helps when solving \kLinkage/, where the task is to connect terminal pairs by disjoint paths, \cite{johnson2001directed}, the algorithm
has a running time of the form $\Bo(n^{f(k, dtw)})$, where $k$ is the number of terminals and $dtw$ is the directed treewidth of the input digraph.
At the same time, there is no $f(k)n^{g(dtw)}$-time algorithm for \kLinkage/ \cite{slivkins2010parameterized} under standard assumptions, and
many further problems remain hard even if the directed treewidth of the input is a constant \cite{GanianHKLOR14}.

One of the shortcomings of directed treewidth is that it cannot explain the structural complexity of acyclic digraphs,
as those digraphs have directed treewidth zero.
Indeed, the digraph constructed in the hardness reduction for \kLinkage/ provided by \cite{slivkins2010parameterized} is acyclic.
Since fundamental problems like \kLinkage/ remain \NP/-hard even if the input digraph is acyclic, it is natural to search for additional parameters which may help in the study of the structural of digraphs and also of acyclic digraphs.

Funnels are an algorithmically useful subclass of acyclic digraphs \cite{milani2020efficient}.
For example, it is easy to solve \kLinkage/ in polynomial time on funnels.
Further, \FADSlong/, the problem of deleting at most $k$ arcs from a digraph in order to obtain a funnel, admits a polynomial kernel \cite{milani2022polynomial}.

Funnels have found application in \emph{RNA assembly} \cite{KhanKCWT22},
modeling a class of digraph on which \ProblemName{Flow Decomposition}
is easy to solve.
Additionally, \cite{Caceres23} considers two generalization of funnels,
namely \(k\)-\emph{funnels} and
a class called \(\mathcal{ST}_k\), and
then shows that 
\ProblemName{String Matching to Labeled Graph} can be solved 
more efficiently on \(k\)-funnels and in digraphs of the class \(\mathcal{ST}_k\)
if \(k\) is small.

In this work, we generalize the properties of funnels by defining a parameter called \emph{ear anonymity}.
This parameter is defined in such a way that funnels are exactly the acyclic digraphs with ear anonymity one.
We show that, while computing the ear anonymity of a digraph is \NP/-hard in general,
it can be computed in $\Bo(m(n + m))$-time if the input digraph is acyclic.

We define ear anonymity together with three relevant computational problems in \cref{sec:def}.
In \cref{sec:ear-anonymity-on-dags} we prove one of our main results, providing a polynomial-time algorithm for \DPA/ on acyclic digraphs.
In \cref{sec:eis-is-np-hard}, we show that all computational problems defined here regarding ear anonymity are \NP/-hard in the general setting.
Further, in \cref{sec:sigma-2-p-hardness-pathpa} we show another of our main results, namely that one of these problems is even $\Sigma_2^p$-complete, a class which is ``above'' \NP/ in the polynomial hierarchy.
To achieve this, we define two additional computational problems which we use to help us construct our reduction, proving that each of them is also $\Sigma_2^p$-hard.

Since the literature on hardness results on higher levels of the polynomial hierarchy is not as rich as for \NP/-hardness results,
we consider the techniques used in \cref{sec:sigma-2-p-hardness-pathpa} to be of independent interest and to be potentially useful in showing that further problems on digraphs are $\Sigma_2^p$-complete.
In particular, the auxiliary problems considered are related to finding linkages in directed graphs, a fundamental problem often used in \NP/-hardness reductions.
To the best of our knowledge (see \cite{schaefer2002completeness} for a survey on related hardness results), none of the hard problems in the polynomial hierarchy ``above'' \NP/ studied so far are related to linkages on digraphs.

Finally, we provide some concluding remarks and discuss future work in \cref{sec:ear-anonymity:remarks}.
\ifappendix
Due to space constrains, proofs marked with (\appsymb{}) are deferred to the appendix.
\fi

\section{Preliminaries}

A \emph{directed graph}, or \emph{digraph}, is a tuple $D \coloneqq \Brace{V, E}$ where $V$ is the \emph{vertex set} and $E \subseteq \Set{(v,u) \mid v,u \in V \text{ and } v \neq u}$ is the \emph{arc set}.
We write $\V{D}$ for the set $V$ and $\A{D}$ for the set $E$.

The inneighbors of a vertex $v$ in a digraph $D$ are denoted by $\InN[D]{v} = \{ u \in V \mid (u,v) \in E\}$; its outneighbors are given by $\OutN[D]{v} = \{ u \in V \mid (v,u) \in E\}$.
The indegree of $v$ is written as $\In[D]{v} = \Abs{\InN[D]{v}}$, and its outdegree as $\Out[D]{v} = \Abs{\OutN[D]{v}}$.
A vertex \(v\) is a \emph{source} if \(\In[D]{v} = 0\) and
a \emph{sink} if \(\Out[D]{v} = 0\).
We omit the index \(D\) if the digraph is clear from the context.

We extend the definition of set operators for digraphs.
 Let~$G = (V,E), H = (U, F)$ be two digraphs.
We define
\begin{align*}
	H \subseteq G & \Leftrightarrow U \subseteq V \text{ and } F \subseteq E,\\
	H \cup G & = (U \cup V, F \cup E),\\
	H \cap G & = (U \cap V, F \cap E),\\
	H \setminus G & = H - \V{G}.
\end{align*}
In particular, we say that~$H$ is a \emph{subgraph} of~$G$ if $H \subseteq G$.

A \emph{walk} of length $\ell$ in $D$ is a vertex sequence $W \coloneqq \Brace{v_0, v_1, \dots, v_{\ell}}$ such that $\Brace{v_i, v_{i+1}} \subseteq \A{D}$ holds for all $0 \leq i < \ell$.
We say that $W$ is a $v_0$-$v_\ell$-walk and
write $\Start{W}$ for $v_0$ and $\End{W}$ for $v_\ell$.

A walk $W$ is said to be a $v_0$-$v_\ell$-\emph{path} if no vertex appears twice along the walk;
$W$ is a cycle if $v_0 = v_\ell$ and $(v_0, v_1, \dots, v_{\ell - 1})$ is a path and $\ell \geq 2$;
further, $W$ is a \emph{directed ear} if it is either a path or a cycle.
Finally, $D$ is \emph{acyclic} if it does not contain any cycles.

Given two walks $W_1 = (x_1, x_2, \dots, x_{j}), W_2 = (y_1, y_2, \dots, y_{k})$ with $\End{W_1} = \Start{W_2}$, we make use of the concatenation notation for sequences and write $W_1 \cdot W_2$ for the walk $W_3 \coloneqq (x_1, x_2, \dots, x_{j}, y_2,\linebreak y_3, \dots, y_{k})$.
If $W_1$ or $W_2$ is an empty sequence, then the result of $W_1 \cdot W_2$ is the other walk (or the empty sequence if both walks are empty).

Let $P$ be a path and $X$ a set of vertices with $\V{P} \cap X \neq \emptyset$.
We consider the vertices $p_1,\dots,p_m$ of $P$ ordered by their occurrence on $P$.
Let $i$ be the highest index such that $p_i \in X$, we call $p_i$ the \emph{last vertex of $P$ in $X$}.
Similarly, for the smallest index $j$ with $p_j \in X$ we call $p_j$ the \emph{first vertex of $P$ in $X$}.

In digraphs, the vertices which can be reached from a vertex~$v$ are given by~$\OutR{v}$.
The vertices which can reach~$v$ are given by~$\InR{v}$.
That is~$u \in \InR{v}$ if and only if there is a \Patht{u}{v}, and~$u \in \OutR{v}$ if and only if there is a \Patht{v}{u}.

Given a digraph $D$ and an arc $(v,u) \in \A{D}$, we say that $(v,u)$ is \emph{butterfly contractible} if $\Out{v} = 1$ or $\In{u} = 1$.
The \emph{butterfly contraction} of $(v,u)$ is the operation which consists of removing $v,u$ from $D$, then adding a new vertex $vu$, together with the arcs $\Set{(w, vu) \mid w \in \InN[D]{v} \setminus \{u\}}$ and $\Set{(vu,w) \mid w \in \OutN[D]{u} \setminus \{v\}}$.
Note that, by definition of digraph, we \emph{remove} duplicated arcs and arcs of the form $(w,w)$.
If there is a subgraph $D'$ of $D$ such that we can construct another digraph $H$ from $D'$ by means of butterfly contractions, then we say that $H$ is a \emph{butterfly minor of $D$}, or that \emph{$D$ contains $H$ as a butterfly minor}.

A \emph{subdivision} of an arc $(v,u)$ is the operation of replacing $(v,u)$ by a path $v,w,u$.
We say that digraph $H$ is a \emph{topological minor} of a digraph $D$ if some subdivision $H'$ of $H$, obtained by iteratively subdiving arcs, is isomorphic to some subgraph of $D$. 
See \cite{diestel2017,Gutin-Digraphs-2008} for further information on digraphs.

\begin{definition}[\cite{milani2020efficient}]
	\label{def:funnel}
	A digraph~$D$ is a \emph{funnel} if~$D$ is a DAG and for every path~$P$ from a source to a sink of~$D$ of length at least one there is some arc $a \in A(P)$ such that for any different path~$Q$ from a (possibly different) source to a (possibly different) sink we have $a \not\in A(Q)$.
\end{definition}

Given two sets \(A, B\)	of vertices in a digraph \(D\),
we say that a set of pairwise vertex-disjoint paths \(\mathcal{L}\) is a \emph{linkage} from \(A\) to \(B\) 
if all paths in \(\mathcal{L}\)	start in \(A\) and end in \(B\).

\begin{definition}
	\label{def:sigma-i-p}
	Let $i \geq 1$.
	A language $L$ over an alphabet $\Gamma$ is in $\Sigma_i^p$ if there exists a polynomial $p : \N \rightarrow \N$ and a polynomial-time Turing Machine $M$ such that for every word $x \in \Gamma^*$,
	\begin{align*}
		x \in L \Leftrightarrow \exists u_1 \in \Gamma^{p(\Abs{x})}\, \forall u_2 \in \Gamma^{p(\Abs{x})} \ldots Q_i u_i \in \Gamma^{p(\Abs{x})}\, M(x,u_1, \ldots, u_i) = 1,
	\end{align*}
	where $Q_i$ denotes $\forall$ or $\exists$ depending on whether $i$ is even or odd, respectively.

	The \emph{polynomial hierarchy} is the set $\PH/ = \bigcup_i \Sigma_i^p$.
\end{definition}

It is easy to verify that $\Sigma_1^p = \NP/$.
We also define, for every $i \geq 1$, the class $\Pi_i^p = \textsf{co}\Sigma_i^p = \Set{\overline{L} \mid L \in \Sigma_i^p}$.
The concept of reductions and hardness can be defined in a similar way as for \NP/-completeness.
See \cite{arora2009computational} for further information on computational complexity.

\section{Definition of Ear Anonymity}
\label{sec:def}

Funnels are characterized by how easy it is to uniquely identify a maximal path\footnote{Maximal with respect to the subgraph relation.}: it suffices to take the private arc of the path.
In acyclic digraphs, the maximal paths correspond exactly to the paths which start in a source and end in a sink.
In general, this does not have to be case.
Indeed, a cycle contains several distinct maximal paths, all of them overlapping.
Hence, it is natural that, in general digraphs, we consider not only how to identify maximal paths, but also cycles, leading us to the well-known concept of ears.

We then come to the question of how to uniquely identify a maximal ear in a digraph.
Clearly, a single arc does not always suffice, as it can be in several ears.
If we take a set of arcs, ignoring their order on the ear, then some rather simple digraphs will require a large number of arcs to uniquely identify an ear, for example the digraph in \cref{fig:ear-identifying-sequence-vs-set}.
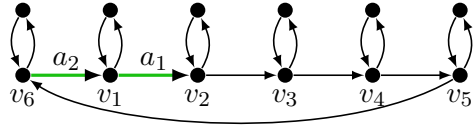
\begin{figure}[H]
	\centering
			\centering
	
		\begin{tikzpicture}[yscale=0.75]
		\node[circle, label = below:{{$v_1$}}, line width = 0.599, fill = black]
	(n0) at (1.15, 0){};
\node[circle, line width = 0.599, fill = black]
	(n4) at (1.15, 1.15){};
\node[circle, line width = 0.599, fill = black]
	(n5) at (4.6, 1.15){};
\node[circle, line width = 0.599, fill = black]
	(n6) at (5.75, 1.15){};
\node[circle, line width = 0.599, fill = black]
	(n7) at (3.45, 1.15){};
\node[circle, line width = 0.599, fill = black]
	(n8) at (2.3, 1.15){};
\node[circle, line width = 0.599, fill = black]
	(n9) at (0, 1.15){};
\node[circle, label = below:{{$v_2$}}, line width = 0.599, fill = black]
	(n10) at (2.3, 0){};
\node[circle, label = below:{{$v_3$}}, line width = 0.599, fill = black]
	(n11) at (3.45, 0){};
\node[circle, label = below:{{$v_4$}}, line width = 0.599, fill = black]
	(n1) at (4.6, 0){};
\node[circle, label = below:{{$v_5$}}, line width = 0.599, fill = black]
	(n2) at (5.75, 0){};
\node[circle, label = below:{{$v_6$}}, line width = 0.599, fill = black]
	(n3) at (0, 0){};
\path[-latex, line width = 1.2, draw = green]
	(n3) to (n0);
\path[-latex, line width = 1.2, draw = green]
	(n0) to (n10);
\path[-latex, line width = 0.599, draw = black]
	(n10) to (n11);
\path[-latex, line width = 0.599, draw = black]
	(n11) to (n1);
\path[-latex, line width = 0.599, draw = black]
	(n1) to (n2);
\path[-latex, line width = 0.599, draw = black]
	(n3) .. controls (0.212, 0.485) and (0.139, 0.72) .. (n9);
\path[-latex, line width = 0.599, draw = black]
	(n9) .. controls (-0.212, 0.665) and (-0.139, 0.43) .. (n3);
\path[-latex, line width = 0.599, draw = black]
	(n4) .. controls (0.938, 0.665) and (1.011, 0.43) .. (n0);
\path[-latex, line width = 0.599, draw = black]
	(n8) .. controls (2.088, 0.665) and (2.161, 0.43) .. (n10);
\path[-latex, line width = 0.599, draw = black]
	(n7) .. controls (3.238, 0.665) and (3.311, 0.43) .. (n11);
\path[-latex, line width = 0.599, draw = black]
	(n5) .. controls (4.388, 0.665) and (4.461, 0.43) .. (n1);
\path[-latex, line width = 0.599, draw = black]
	(n6) .. controls (5.538, 0.665) and (5.611, 0.43) .. (n2);
\path[-latex, line width = 0.599, draw = black]
	(n0) .. controls (1.362, 0.485) and (1.289, 0.72) .. (n4);
\path[-latex, line width = 0.599, draw = black]
	(n10) .. controls (2.512, 0.485) and (2.439, 0.72) .. (n8);
\path[-latex, line width = 0.599, draw = black]
	(n11) .. controls (3.662, 0.485) and (3.589, 0.72) .. (n7);
\path[-latex, line width = 0.599, draw = black]
	(n1) .. controls (4.812, 0.485) and (4.739, 0.72) .. (n5);
\path[-latex, line width = 0.599, draw = black]
	(n2) .. controls (5.962, 0.485) and (5.889, 0.72) .. (n6);
\path[-latex, line width = 0.599, draw = black]
	(n2) .. controls (4.151, -1.254) and (1.01, -0.917) .. (n3);

		\path[]
			(n3) to node[above] {{$a_2$}} (n0);
		\path[]
			(n0) to node[above] {{$a_1$}} (n10);
	\end{tikzpicture}
		\vspace{-0.25cm}
				\caption{For any subset of at most 5 arcs of the cycle $\Brace{v_1, v_2, v_3, v_4, v_5, v_6, v_1}$ we can find some path visiting such arcs which is distinct from the cycle considered.}
		\label{fig:ear-identifying-sequence-vs-set}
	\end{figure}
Hence, we consider not only the arcs of the ear, but also their order along the ear.
We also require the existence of at least one arc in the identifying sequence in order to ensure the parameter is closed under the subgraph relation.
\begin{definition}
	\label{def:identifying-sequence}
	Let $P$ be an ear.
	A sequence $\Brace{a_1, a_2, \dots, a_{k}}$ of arcs of $P$
	is an \emph{identifying sequence for} $P$
	if \(k \geq 1\) and
	every ear $Q$ containing
	$\Brace{a_1, a_2, \dots, a_{k}}$ in this order
	is a subgraph of $P$.
\end{definition}
Note that the cycle $\Brace{v_1, v_2, v_3, v_4, v_5, v_6, v_1}$ in \cref{fig:ear-identifying-sequence-vs-set} admits an identifying sequence of length two, namely $\bar{a} = \Brace{a_1, a_2}$ where $a_1 = \Brace{v_1, v_2}$ and $a_2 = \Brace{v_6, v_1}$.
We also note that using vertices instead of arcs does not lead to a well-defined parameter, as can be observed in the example given in \cref{fig:arc-sequence-vs-vertex-sequence}.
As every ear can be uniquely described by ordering its entire arc-set according to its occurrence along the ear, the parameter defined above is well-defined for all ears.

\begin{figure}[H]
	\centering
			
		\begin{tikzpicture}[yscale=0.6]
		\node[circle, label = left:{{$v_1$}}, line width = 0.9, fill = black]
	(v1) at (0, 0){};
\node[circle, label = right:{{$v_3$}}, line width = 0.9, fill = black]
	(v2) at (2.3, 0){};
\node[circle, label = above:{{$v_2$}}, line width = 0.9, fill = black]
	(v3) at (1.15, 1.15){};
\path[-latex, line width = 0.9, draw = black]
	(v1) to (v2);
\path[-latex, line width = 0.9, draw = black]
	(v1) to (v3);
\path[-latex, line width = 0.9, draw = black]
	(v3) to (v2);

	\end{tikzpicture}		
			\caption{A digraph with two maximal ears.
		While the ear $\Brace{v_1, v_2, v_3}$ contains all vertices of the ear $\Brace{v_1, v_3}$, both ears admit an identifying sequence of length 1.
	}
	\label{fig:arc-sequence-vs-vertex-sequence}
	\end{figure}
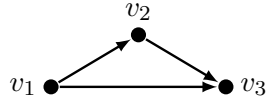

Intuitively, the shorter the identifying sequence of an ear, the less information is necessary in order to uniquely identify or find such an ear.
This leads to the following definition (we ignore maximal ears consisting of a single vertex, as they are seldom interesting and can be found in linear time).
\begin{definition}
	\label{def:ear-ear-anonymity}
	Let $P = \Brace{a_1, a_2, \ldots, a_{k}}$ be a maximal ear in a digraph $D$, given by its arc-sequence (in the case of a cycle, any arc of $P$ can be chosen as $a_1$).
	The \emph{ear anonymity} of $P$ in $D$, denoted by $\Anon[D](P)$, is the length of the shortest identifying sequence for $P$.
	If $k = 0$, we say that $\Anon[D](P) = 0$.
\end{definition}

As we are often interested in the worst-case running time of an algorithm, if some ear of a digraph has high anonymity, then this digraph could be a difficult instance in the worst case.
\begin{definition}
	\label{def:digraph-ear-anonymity}
	The ear anonymity of a digraph $D$, denoted by $\Anon(D)$, is the maximum ear anonymity of the maximal ears of $D$.	
\end{definition}

It is a simple exercise to compare \cref{def:digraph-ear-anonymity} and
\cref{def:funnel} to verify the following observation.
\begin{observation}
	\label{obs:funnel-path-anonymity-one}
	An acyclic digraph $D$ is a funnel if and only if $\Anon(D) \leq 1$.
\end{observation}

	It is sometimes useful to know that a parameter is closed under certain operations.
	For ear anonymity, we can show the following.

		\begin{observation}
			\label{state:ear-anonymity-minor-closed}
		Let $D, H$ be digraphs such that $H$ is a butterfly minor or topological minor of $D$.
		Then $\Anon(H) \leq \Anon(D)$.
	\end{observation}
	
	\begin{proof}
		Since every topological minor of \(D\) is also a butterfly minor of \(D\),
		it suffices to prove the statement for the case when \(H\) is a butterfly minor of \(D\).

		It is immediate from the definition that the statement holds if $H$ is a subgraph of $D$,
		as any conflicting ear for a pair $(P, \bar{a})$ in $H$ would also be a conflicting ear for the same pair in $D$.

		By induction, it suffices to consider the case where
		$H$ is obtained by butterfly contracting an arc $(v,u)$ of $D$ into a vertex $vu$.
		Without loss of generality, $\Out[D]{v} = 1$.
		The case where $\In[D]{u} = 1$ follows analogously.

		Let $P$ be a maximal ear in $H$.
		If $P$ does not contain the vertex $vu$,
		then it is also a maximal ear in $D$, and
		an ear-identifying sequence $\bar{a}$ for $P$ in $D$
		is also an ear-identifying sequence for $P$ in $H$, because
		any conflicting ear $Q$ for $(P, \bar{a})$ in $H$ corresponds to a conflicting ear for $(P, \bar{a})$ in $D$,
		where the vertex $vu$ is replaced by $(v,u)$ in $Q$.

		If $P$ contains $vu$,
		we distinguish between two cases.

		\begin{CaseDistinction}
			\Case \(vu = \Start{P}\).

			Let \(P'\) be the path in \(D\) obtained by replacing \(vu\) with \((v,u)\) in \(P\).
			Clearly, \(P'\) is also a maximal ear.
			Let \(\bar{a}\) be an ear-identifying sequence for \(P'\) in \(D\).

			If \(\bar{a}\) does not contain \((v,u)\), then it is also an ear-identifying sequence for \(P\) in \(H\).
			Otherwise, let \(w\) be the successor o \(vu\) along \(P\).
			Construct a sequence \(\bar{b}\) by replacing \((v,u)\) with \((vu,w)\) in \(\bar{a}\) (or just deleting \((v,u)\) if \((u,w)\) is already in \(\bar{a}\)).

			Assume that there is a conflicting ear \(Q\) for \(P', \bar{b}\) in \(H\).
			Then \(Q\) must contain \(vu\).
			Thus, the ear \(Q'\) obtained by replacing \(vu\) with \((v,u)\) is a conflicting ear for \((P, \bar{a})\), a contradiction.

			\Case \(vu\) has a predecessor \(w\) along \(P\).

			\begin{Subcase}
				\Case \(w\) is an inneighbor of \(v\) in \(D\).

				Let \(P'\) be the path obtained by replacing \(vu\) with \((v,u)\) in \(P\).
				Let \(\bar{a}\) be an ear-identifying sequence for \(P'\) in \(D\).

				Let $\bar{b}$ be the sequence obtained by replacing $(v,u)$ with $(w,vu)$.
				Since containing $(w,vu)$ in $H$ is equivalent to containing $(w,v)$ and $(v,u)$ in $D$,
				the sequence $\bar{b}$ is an ear-identifying sequence for $P$ in $H$.

				\Case \(w\) is not an inneighbor of \(v\) in \(D\).
					
				Then \(w\) must be an inneighbor of \(u\) in \(D\). 
				We define \(P'\) as the path obtained by replacing \(vu\) with \(u\) in \(P\).
				Let \(\bar{a}\) be an ear-identifying sequence for \(P'\) in \(D\).
				Let \(\bar{b}\) be the sequence obtained by replacing all (if any)
				occurrences of \(u\) in \(\bar{a}\) with \(vu\).

				Assume towards a contradiction that there is a conflicting ear \(Q\)
				for \((P, \bar{b})\) in \(H\).
				If \(Q\) does not contain \(vu\),
				then it is also a conflicting ear for \((P', \bar{a})\) in \(D\), a contradiction.

				Let \(Q_v\) be the ear obtained by replacing \(vu\) with \((v,u)\) in \(Q\), and
				let \(Q_u\) be the ear obtained by replacing \(vu\) with \(u\) in \(Q\).
				At least one of \(Q_v, Q_u\) is a valid ear in \(D\), as \(\Out[D]{v} = 1\).

				All arcs in \(\bar{a}\) which do not contain \(u\) are both in \(Q\) and in
				\(Q_v, Q_u\).
				Further, \(Q_v, Q_u\) can only avoid
				an arc of \(\bar{a}\) which contains \(u\) if
				the corresponding arc containing \(vu\) is missing in \(\bar{b}\),
				which cannot be the case by construction of \(\bar{b}\).
				Hence, one of \(Q_v, Q_u\) is a conflicting ear for \((P', \bar{a})\) in \(D\).
			\end{Subcase}
		\end{CaseDistinction}

		We conclude that $\Anon(H) \leq \Anon(D)$, as desired.
	\end{proof}

	We will use \cref{state:ear-anonymity-minor-closed} later to draw a connection between ear anonymity and directed treewidth.
	We now investigate the complexity of computing the ear anonymity of a digraph.
	\Cref{def:identifying-sequence,def:ear-ear-anonymity,def:digraph-ear-anonymity} naturally lead us to three related computational problems.

	Since most of the literature on decision problems concerns itself with problems in \NP/, we formulate the question of our decision problems as an ``existential'' question (instead of a ``for all'' question).
	Hence, the question of whether an arc-sequence $\bar{a}$ is an identifying sequence for an ear $P$ becomes the question of the existence of another ear as defined below.
	\begin{definition}
		\label{def:conflicting-ear}
		Let $P$ be an ear and let $\bar{a}$ be a sequence of arcs of $P$, sorted according to their order on $P$.
		We say that an ear $Q$ is a \emph{conflicting ear for $\Brace{P, \bar{a}}$} if $Q$ visits the arcs of $\bar{a}$ in the given order, yet $Q$ is not a subgraph of $P$.
	\end{definition}

	It is immediate from definition that a sequence $\bar{a}$ is an identifying sequence for an ear if, and only if, no conflicting ear exists.
	The first problem we consider can then be formulated as follows.
	
	\EarISDef

	From the above definition it is trivial to derive the following observation.
	
	\begin{observation}
		\label{state:eis-is-in-np}
		\EarIS/ is in \NP/.
	\end{observation}
						
	Note that the question ``is $\bar{a}$ an identifying sequence for $P$?'' is the complement of \EarIS/ and thus, by \cref{state:eis-is-in-np}, a \coNP/ question.
	We can also formulate this question as ``for all ears $Q$, is $Q$ not a conflicting ear for $\Brace{P, \bar{a}}$?''.
	When considering the problem of determining the ear anonymity of an ear, it seems thus unavoidable to have a quantifier alternation in the question: asking for the existence of an identifying sequence means chaining an existential question with a ``for all'' question.
	
	\PathPADef

	Unlike \EarIS/, it is not clear from the definition whether \PathPA/ is in \NP/, but one can easily verify containment in a class higher up in the polynomial hierarchy.

	\begin{observation}
		\label{state:ear-ear-anonymity-in-sigma-2-p}
		\PathPA/ is in $\Sigma_2^p$.
	\end{observation}
						
	As before, asking if an ear has high anonymity is equivalent to asking if no short identifying sequence for that ear exists.
	It seems again unavoidable to add another quantifier alternation when deciding if a digraph has high ear anonymity: asking if a digraph has high ear anonymity means asking for the existence of an ear for which no short identifying sequence exists.

	\DPADef	

	While it is not clear from the definition whether \DPA/ is even in $\Sigma_2^p$, it is easy to verify that it is in $\Sigma_3^p$.

	\begin{observation}
		\label{state:ear-anonymity-in-sigma-3-p}
		\DPA/ is in $\Sigma_3^p$.
	\end{observation}
					
	In \cref{sec:ear-anonymity-on-dags} we show that \EarIS/, \PathPA/ and \DPA/ are in \P/ on DAGs.
	In \cref{sec:eis-is-np-hard} we show that the three previous decision problems are \NP/-hard in general using some of the results from \cref{sec:ear-anonymity-on-dags}.
	Finally, in \cref{sec:sigma-2-p-hardness-pathpa}, we show that \PathPA/ is $\Sigma_2^p$-complete.

\section{\PAlong/ on DAGs}
\label{sec:ear-anonymity-on-dags}

	We start by identifying certain substructures which increase the anonymity of an ear by enforcing certain arcs to be present in any identifying sequence.
	Two such substructures, called \emph{deviations} and \emph{bypasses},
	are defined in \cref{def:deviation,def:forward-shortcut} and illustrated in \cref{fig:ear-anonymity:forward-shortcut,fig:ear-anonymity:derail} below.
	Of particular interest are subpaths of an ear which must be \emph{hit} by any identifying sequence.
	We call these subpaths \emph{blocking subpaths} since they prevent a potential conflicting ear from containing the corresponding bypass or deviation as a subgraph.

	\begin{definition}
		\label{def:deviation}
		Let $P$ be an ear and let $Q$ be a path in a digraph $D$.
		We say that $Q$ is a \emph{deviation for $P$} if $Q$ is internally disjoint from $P$ and exactly one of $\End{Q}, \Start{Q}$ lies in $P$.
		Additionally, the $\Start{P}$-$\End{Q}$ subpath of $P$ is called a \emph{blocking subpath for $Q$} if $\End{Q} \in \V{P}$, and the $\Start{Q}$-$\End{P}$ subpath of $P$ is called a \emph{blocking subpath for $Q$} if $\Start{Q} \in \V{P}$.
	\end{definition}

	\vspace{-0.5cm}
\begin{figure}[H]
	\centering
				
		\begin{tikzpicture}[yscale=0.7]
		 \node[circle, label = below:{{$v_1$}}, line width = 0.6, fill = black]
	(n0) at (0, 0){};
\node[circle, label = below:{{$v_2$}}, line width = 0.6, fill = black]
	(n1) at (1.15, 0){};
\node[circle, label = below:{{$v_3$}}, line width = 0.6, fill = black]
	(n2) at (2.3, 0){};
\node[circle, label = left:{\color{black}{$u$}}, line width = 0.6, fill = black]
	(n3) at (2.3, 1.15){};
\path[-latex, line width = 0.6, draw = black]
	(n0) to (n1);
\path[-latex, line width = 1.2, draw = green]
	(n1) to (n2);
\path[-latex, line width = 0.6, draw = red]
	(n1) to (n3);

	\end{tikzpicture}
				\caption{The path $(v_2, u)$ is a deviation for
		the path $P = (v_1, v_2, v_3)$.
		The unique identifying sequence of length one for $P$ is
		$((v_2, v_3))$.}
		\label{fig:ear-anonymity:derail}
	\end{figure}
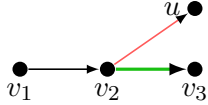

	\begin{definition}
		\label{def:forward-shortcut}
		Let $P$ be a path in a digraph $D$.
		Let $v_1, v_2, \dots v_{n}$ be the vertices of $P$ sorted according to their order in $P$.
		A \emph{bypass for $P$} is a path $Q$ in $D$ from some $v_i$ to some $v_j$ with $i < j$
		such that
		$V(Q) \cap V(P) = \{v_i, v_j\}$ and
		\(Q\) is not a subpath of \(P\).
		Further, the $v_i$-$v_j$ subpath of $P$ is called the \emph{blocking subpath for $Q$}.
	\end{definition}

	\vspace{-0.5cm}
	\begin{figure}[H]
		\centering
						
			\begin{tikzpicture}[yscale=0.7]
			\node[circle, label = below:{{$v_1$}}, line width = 0.6, fill = black, draw = black]
	(v1) at (0, 0){};
\node[circle, label = above:{{$u$}}, line width = 0.6, fill = black, draw = black]
	(v2) at (2.3, 1.15){};
\node[circle, label = below:{{$v_2$}}, line width = 0.6, fill = black, draw = black]
	(v3) at (1.15, 0){};
\node[circle, label = below:{{$v_3$}}, line width = 0.6, fill = black, draw = black]
	(v4) at (3.45, 0){};
\node[circle, label = below:{{$v_4$}}, line width = 0.6, fill = black, draw = black]
	(v5) at (4.6, 0){};
\path[-latex, line width = 1.2, draw = green]
	(v3) to (v4);
\path[-latex, line width = 0.6, draw = black]
	(v1) to (v3);
\path[-latex, line width = 0.6, draw = black]
	(v4) to (v5);
\path[-latex, line width = 0.6, draw = red]
	(v3) to (v2);
\path[-latex, line width = 0.6, draw = red]
	(v2) to (v4);

		\end{tikzpicture}
								\caption{The path $\Brace{v_2, u, v_3}$ is a bypass for $P = \Brace{v_1, v_2, v_3, v_4}$.
		Note that there is exactly one identifying sequence of length 1 for $P$, namely $\Brace{\Brace{v_2, v_3}}$.}
		\label{fig:ear-anonymity:forward-shortcut}
			\end{figure}
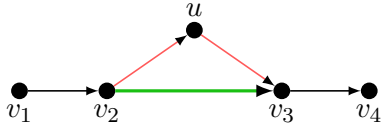

	If an ear contains many arc-disjoint blocking subpaths, then every identifying sequence must be long.
	If, on the other hand, the blocking subpaths overlap, then a short identifying sequence may still exist.
	In order to better analyze the relationship between the length of an identifying sequence and the blocking subpaths of an ear, we model this problem as a problem on intervals.
	Intuitively, we can consider each arc on an ear to be an integer, ordered naturally along the ear, and each blocking subpath as an interval over the integers.
	Hence, we are interested in finding a minimum set of integers which hit all the intervals.
	This naturally leads us to the definitions given below.

	\begin{definition}
		\label{def:blocking-interval-digraph}
		Let $Q_1, Q_2, \dots, Q_{k}$ be subpaths of an ear $P$.
		The \emph{arc-interval set of $Q_1, Q_2, \dots, Q_{k}$} is the set of intervals $\mathcal{I} = \Set{I_1, I_2, \dots, I_{k}}$ with $I_i = \A{Q_i}$ for all $1 \leq i \leq k$.
	\end{definition}

	\begin{definition}
		\label{def:interval-hitting-set}
		Let $\mathcal{I} = \Set{I_1, I_2, \dots, I_{n}}$ be a set of intervals over a finite (ordered) domain $U$.
		A set $X \subseteq U$ is a \emph{hitting set} for $\mathcal{I}$ if $I_i \cap X \neq \emptyset$ for every $I_i \in \mathcal{I}$.
	\end{definition}

	Since an ear can have an exponential number of bypasses and deviations, we are interested in reducing the number of blocking subpaths we need to consider.
	In particular, if a blocking subpath is fully contained within another, then we can ignore the longer subpath.

	Formally, we define a partial ordering $\preceq$ over the blocking paths of the bypasses and the deviations for an ear $P$ as follows.
	For two blocking subpaths $B_a, B_b$ set $B_a \preceq B_b$ if $\Start{B_a}$ is not before $\Start{B_b}$ in $P$ and $\End{B_a}$ is not after $\End{B_b}$ in $P$.
	That is, \(B_a \preceq B_b\) if and only if \(B_a\) is a subpath of \(B_b\).

	Let $B_1, B_2, \dots B_{k}$ be the minimal elements of $\preceq$.
	Every set of intervals $\mathcal{I}$ which contains the arc-interval set of each $B_1, B_2, \dots B_{k}$ is called the \emph{blocking interval set for $P$}.
	If \(\mathcal{I}\) contains only the arc-interval sets of $B_1, B_2, \dots B_{k}$, then it is
	the minimum blocking interval set for \(P\).

	We now establish a connection between hitting sets for a blocking interval set for an ear and the identifying sequence for that ear.

	\begin{lemma}
		\label{state:hitting-set-implies-identifying-sequence}
		Let $P$ be a maximal ear in an acyclic digraph $D$ and 
		let $\mathcal{I}$ be a blocking interval set for $P$.
		Let $\bar{a} = \Brace{a_1, a_2, \dots, a_{k}}$ be a hitting set for $\mathcal{I}$, sorted according to the occurrence of the arcs along $P$.
		Then $\bar{a}$ is an identifying sequence for $P$.
	\end{lemma}
	\begin{proof}
		Let $Q$ be a maximal ear visiting $\bar{a}$ in this order.
		Let $(v^s_i, v^e_i) = a_i$ for each $1 \leq i \leq k$.
		Partition $Q$ and $P$ as follows.
		For each $2 \leq i \leq k$ let $Q_i$ be the $v^e_{i-1}$-$v^e_i$ subpath of $Q$ and let $P_i$ be the $v^e_{i-1}$-$v^e_i$ subpath of $P$.
		Let $Q_1$ be the $\Start{Q}$-$v^e_1$ subpath of $Q$ and let $Q_{k+1}$ be the $v^e_k$-$\End{Q}$ subpath of $Q$.
		Similarly, let $P_1$ be the $\Start{P}$-$v^e_1$ subpath of $P$ and let $P_{k+1}$ be the $v^e_k$-$\End{P}$ subpath of $P$.
		Note that $P = P_1 \cdot P_2 \cdot \ldots \cdot P_{k+1}$ and $Q = Q_1 \cdot Q_2 \cdot \ldots \cdot Q_{k+1}$.

		Assume towards a contradiction that $Q \neq P$.
	  In particular, \(Q_i \neq P_i\) holds for some $1 \leq i \leq k + 1$.
	  
		If \(Q_i\) contains an arc \(v_i, v_j\)
		such that
		both \(v_i\) and \(v_j\) lie in \(P\),
		but \((v_i, v_j)\) is not an arc in \(P\),
		then there is no arc of \(\bar{a}\)
		between \(v_i\) and \(v_j\) along \(P\), as \(D\) is acyclic.
		However, \((v_i, v_j)\) is a bypass for \(P\), and
		its corresponding blocking subpath is not hit by \(\bar{a}\), a contradiction to the choice of \(\bar{a}\).
		Hence, \(Q_i\) must contain some arc \((u_1, u_2)\)
		such that
		exactly of \(u_1,u_2\) is in \(P\).

		We now distinguish between two cases.

		\textbf{Case 1:} $i = 1$ or $i = k+1$.
		Assume without loss of generality that $i = 1$.
		The case $i = k+1$ follows analogously.
		Let $\Brace{u_1, u_2}$ be the first arc along $Q_1$ such that $u_1 \not\in \V{P_1}$ and $u_2 \in \V{P_1}$.
		Since $\Brace{v^s_1, v_1^e}$ is both in $Q_1$ and in $P_1$ and $Q_1 \neq P_1$, such an arc $\Brace{u_1, u_2}$ exists.

		If $u_2$ comes after or at $v^e_1$ along $P_1$, then $Q_1$ must contain a subpath $Q'$ from $u_2$ to some $u_3$ such that $u_3$ comes before or at $v^s_1$ along \(P_1\).
		This however implies the existence of a cycle in $D$, a contradiction to the assumption that $D$ is acyclic.
		Hence, $u_2$ lies before or at $v^s_1$ along $P_1$.
		
		By definition, $\Brace{u_1, u_2}$ is a deviation for $P_1$.
		Hence, the $\Start{P_1}$-$u_2$ subpath of $P_1$ contains a blocking path $B$ which is not hit by $\bar{a}$, a contradiction to the assumption that $\bar{a}$ is a hitting set for $\mathcal{I}$.

		\textbf{Case 2:} $2 \leq i \leq k$.
		Let $\Brace{u_1, u_2}$ be the first arc along $Q_i$ such that $u_1 \in \V{P_i}$ and $u_2 \not\in \V{P_i}$.
		Since both $Q_i$ and $P_i$ contain $v^e_{i-1}$ and $v^s_i$, such an arc $\Brace{u_1, u_2}$ exists.
		As $\Brace{v^s_i, v^e_i}$ is the last arc along $Q_i$, there must be a $u_2$-$v^s_i$ path $Q'$ in $Q$.
		If $Q'$ intersects some vertex of $P$ which comes at or after $v^e_i$ along $P$, then there is a cycle in $D$, a contradiction.
		Hence, $Q'$ must contain a bypass whose blocking subpath $B$ does not contain any arc of $\bar{a}$, contradicting the assumption that $\bar{a}$ is a hitting set for $\mathcal{I}$.

		As both \textbf{Case 1} and \textbf{Case 2} lead to a contradiction, we conclude that $Q = P$ and, hence, $\bar{a}$ is an identifying sequence for $P$, as desired.
	\end{proof}

	Note that \cref{state:hitting-set-implies-identifying-sequence} is not true if we allow the digraph to contain cycles, with \cref{fig:hitting-set-smaller-than-identifying-sequence} being a counter-example.
	\begin{figure}[H]
		\centering
		\vspace{-1.75cm}
		\begin{tikzpicture}[yscale=0.8]
			\node[circle, label = above:{{$v_3$}}, line width = 0.894, fill = black]
	(v1) at (2.3, 1.15){};
\node[circle, label = above right:{{$v_2$}}, line width = 0.894, fill = black]
	(v2) at (1.15, 1.15){};
\node[circle, line width = 0.894, fill = gray]
	(v3) at (1.15, 0){};
\node[circle, label = above:{{$v_1$}}, line width = 0.894, fill = black]
	(v4) at (0, 1.15){};
\node[circle, label = above:{{$v_7$}}, line width = 0.894, fill = black]
	(v5) at (8.05, 1.15){};
\node[circle, label = above right:{{$v_8$}}, line width = 0.894, fill = black]
	(v6) at (9.2, 1.15){};
\node[circle, line width = 0.894, fill = gray]
	(v7) at (9.2, 0){};
\node[circle, label = above:{{$v_9$}}, line width = 0.894, fill = black]
	(v8) at (10.35, 1.15){};
\node[circle, label = above:{{$v_4$}}, line width = 0.894, fill = black]
	(v9) at (3.45, 1.15){};
\node[circle, label = above:{{$v_5$}}, line width = 0.894, fill = black]
	(v10) at (4.6, 1.15){};
\node[circle, line width = 0.894, fill = gray]
	(v11) at (5.75, 2.3){};
\node[circle, label = above:{{$v_6$}}, line width = 0.894, fill = black]
	(v12) at (6.9, 1.15){};
\path[-latex, line width = 1.2, draw = green]
	(v10) to (v12);
\path[-latex, line width = 0.588, draw = black]
	(v1) to (v9);
\path[-latex, line width = 0.588, draw = black]
	(v9) to (v10);
\path[-latex, line width = 0.588, draw = black]
	(v12) to (v5);
\path[-latex, line width = 0.894, draw = red]
	(v12) .. controls (5.483, -0.303) and (3.684, -0.268) .. (v1);
\path[-latex, line width = 0.894, draw = red]
	(v1) .. controls (4.093, 2.171) and (8.556, 4.342) .. (v6);
\path[-latex, line width = 0.588, draw = black]
	(v2) to (v1);
\path[-latex, line width = 1.2, draw = yellow]
	(v5) to (v6);
\path[-latex, line width = 0.894, draw = red]
	(v2) .. controls (1.167, 3.073) and (5.087, 3.757) .. (v5);
\path[-latex, line width = 1.2, draw = green]
	(v6) to (v8);
\path[-latex, line width = 1.2, draw = green]
	(v4) to (v2);
\path[-latex, line width = 0.894, draw = red]
	(v5) .. controls (8.05, -0.704) and (4.591, -0.02) .. (v9);
\path[-latex, line width = 0.588, draw = gray]
	(v3) to (v2);
\path[-latex, line width = 0.588, draw = gray]
	(v6) to (v7);
\path[-latex, line width = 0.588, draw = gray]
	(v10) to (v11);
\path[-latex, line width = 0.588, draw = gray]
	(v11) to (v12);

		\end{tikzpicture}
		\caption{The set $\Set{(v_1, v_2), (v_5, v_6), (v_8, v_9)}$ is a hitting set of size 3 for the blocking interval set for $P = (v_1, v_2, \ldots, v_9)$, yet $\Anon[D](P) = 4$, witnessed by the sequence
		$((v_1, v_2), (v_5, v_6), (v_7, v_8), (v_8, v_9))$.}
		\label{fig:hitting-set-smaller-than-identifying-sequence}
	\end{figure}
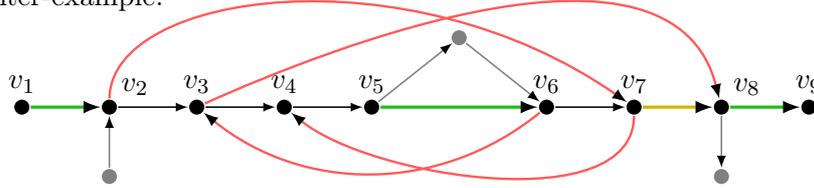
	The reverse direction of \cref{state:hitting-set-implies-identifying-sequence}, however, does hold in general, and is proven below.

		\begin{lemma} 			\label{state:identifying-sequence-implies-hitting-set}
		Let $P$ be a maximal path in a digraph $D$ and let $\mathcal{I}$ be the set of blocking intervals for $P$.
		Let $\bar{a}$ be an identifying sequence for $P$.
		Then $\bar{a}$ is a hitting set for $\mathcal{I}$.
	\end{lemma}
	
	\begin{proof}
		Assume towards a contradiction that $\bar{a}$ is an identifying sequence for $P$ but there is some blocking subpath $B$ in $P$ which does not contain any arc in $\bar{a}$.

		We consider two cases.

		\textbf{Case 1:} $B$ is the blocking subpath of some bypass $S$ in $P$.
		Let $P_a$ be the $\Start{P}$-$\Start{B}$ subpath of $P$ and let $P_b$ be the $\End{B}$-$\End{P}$ subpath of $P$.
		Since $P_a$ and $P_b$ contain all arcs of $P$ except those in $B$, we have that $R \coloneqq P_a \cdot S \cdot P_b$ visits all arcs of $\bar{a}$ in the given order.
		However, $R$ is not a subgraph of $P$, contradicting the assumption that $\bar{a}$ is an identifying sequence for $P$.

		\textbf{Case 2:} $B$ is the blocking subpath of some deviation $S$ of $P$.
		Without loss of generality we assume that $\End{S} \in \V{P}$.
		The case where $\Start{S} \in \V{P}$ follows analogously.
		Let $P_b$ be the $\End{S}$-$\End{P}$ subpath of $P$.

		By definition, $B$ is a $\Start{P}$-$\End{S}$ subpath of $P$.
		Hence, $P_b$ visits all arcs of $\bar{a}$ in the given order.
		Thus, $R \coloneqq S \cdot P_b$ is a path visiting $\bar{a}$, however $R$ is not a subgraph of $P$, a contradiction to the assumption that $\bar{a}$ is an identifying sequence for $P$.

		Since both cases above lead to a contradiction, we conclude that $\bar{a}$ is a hitting set for~$\mathcal{I}$.
	\end{proof}

	Together, \cref{state:hitting-set-implies-identifying-sequence,state:identifying-sequence-implies-hitting-set} allow us to reduce \PathPA/ on acyclic digraphs to a hitting set problem on intervals which can be solved efficiently, as shown below.

	\IntervalHSDef

\begin{varalgorithm}{minimumHittingSet}
	\begin{algorithmic}[1]
		\Function{minimumHittingSet}{set of intervals $\mathcal{I} = \Set{I_0, I_1, \dots, I_{n-1}}$}
				\State $\mathcal{I}_{start} \gets $ sort $\mathcal{I}$ by starting points
		\State $\mathcal{I}_{end} \gets $ sort $\mathcal{I}$ by endpoints
				\State $hit \gets $ an array of length $n$, initialized with \texttt{false}
				\State $i_{start} \gets 0$
		\State $i_{end} \gets 0$
		\State $X \gets \emptyset$
				\While{$i_{end} < n$}
						\State $e \gets \End{\mathcal{I}_{end}[i_{end}]}$
			\State $X \gets X \cup \Set{e}$
						\While{$i_{start} < n \textbf{ and } \Start{\mathcal{I}_{start}[i_{start}]} \leq e$}
			\label{alg:hitting-set:while-start}
				\State $I_j \gets \mathcal{I}_{start}[i_{start}]$
				\State $hit[j] \gets $ \texttt{true}
				\State $i_{start} \gets i_{start} + 1$
			\EndWhile
			\LComment{The function $\FunctionName{id}$ returns the index of the interval, that is, $\FunctionName{id}(I_j) = j$.}
			\While{$i_{end} < n \textbf{ and } hit[\FunctionName{id}(\mathcal{I}_{end}[i_{end}])] = \texttt{true}$} 
			\label{alg:hitting-set:while-end}
				\State $i_{end} \gets i_{end} + 1$
			\EndWhile
		\EndWhile
		\State \Return $X$
		\EndFunction
		\end{algorithmic}
		\caption{Compute a minimum hitting set for a set of intervals.}
		\label{alg:minimum-interval-hitting-set}
\end{varalgorithm}

The algorithm in \cref{state:compute-interval-hitting-set} uses standard techniques (see, for example, \cite{CormenLRS01chap16}) to greedily compute the desired hitting set.

\begin{lemma}
	\label{state:compute-interval-hitting-set}
	An optimal solution for \IntervalHS/ can be computed in $\Bo(n \log n)$ time, where \(n\) is the size of the domain.
	Moreover, if two copies of the interval set \(mathcal{I}\) are provided to the algorithm,
	 whereas one copy is sorted according to the starting points and the other according to the endpoints of the intervals,
	 then the problem can be solved in \(\Bo(n)\) time.
\end{lemma}
\begin{proof}
	Let $\mathcal{I}$ be the input set of intervals.
	Without loss of generality we assume that the domain of the intervals is $\Set{1, 2, \ldots, 2n}$, where $n$ is the number of intervals.
	If this is not the case, we can compress the domain in $\Bo(n \log n)$ time by sorting the start and endpoints and then assigning each one of them a number from 1 to $2n$, preserving the original order.

	We show that the \cref{alg:minimum-interval-hitting-set} computes a minimum hitting set for $I$ in $\Bo(n \log n)$ time.
	Towards this end, let $Y$ be a hitting set of $\mathcal{I}$ that is distinct from $X$.
	Let $e \in X \setminus Y$ be the earliest such element in the domain of $\mathcal{I}$.
	If no such $e$ exists, then the choice of $X$ was clearly optimal and there is nothing to show.

	Otherwise, $e$ was chosen as the endpoint of some interval $I_j$.
	Furthermore, $e$ is the only element in $X$ hitting $I_j$ as all arcs
	added to $X$ afterwards come after the end of $I_j$.
	Hence, there is some $e' \in Y \setminus X$ such that $e'$ also hits $I_j$.
	Further, $e'$ must come before $e$ in the domain of $\mathcal{I}$, as $e$ is the endpoint of $I_j$.
	Let $Y'$ be the elements of $Y$ coming before $e'$ and let $X'$ be the elements of $X$ coming before $e$.

	By assumption, $X' \subseteq Y'$.
	Hence, $Y'$ hits all intervals hit by $X'$ and potentially more.
	Let $\mathcal{I}'$ be the intervals hit by $X'$.
	Since $I_j$ was chosen as the interval with the earliest endpoint in $\mathcal{I} \setminus \mathcal{I}'$, every interval in $\mathcal{I} \setminus \mathcal{I}'$ which is hit by $e'$ is also hit by $e$.
	Thus, the set $Z = (Y \setminus \Set{e'}) \cup \Set{e}$ is a hitting set for $\mathcal{I}$ and $\Abs{Z} \leq \Abs{Y}$.
	By repeatedly applying the argument above, we obtain a hitting set $Z'$ with $X \subseteq Z'$.
	By choosing $Y$ as a minimum hitting set, we obtain that $X$ must be a minimum hitting set as well.

	We now analyze the running time of the algorithm above.
	Sorting the intervals can be done in $\Bo(n \log n)$ time.
The while loops on lines \ref{alg:hitting-set:while-start} and \ref{alg:hitting-set:while-end} iterate at most $n$ times each, as in each iteration the variable $i_{start}$ or $i_{end}$ is incremented by one.
	Further, each iteration takes $\Bo(1)$ time.
	Hence, the total running time is in $\Bo(n \log n)$, as desired.

	If we the set \(\mathcal{I}\) is already sorted both ways during at the input, then we can skip the steps which sort this set,
	obtaining a running time of \(\Bo(n)\) instead.
\end{proof}

In order to effectively use \cref{state:compute-interval-hitting-set} when solving \PathPA/ on acyclic digraphs, we need to be able to efficiently compute the blocking interval set for an ear.

		\begin{lemma}
		\label{state:compute-blocking-interval-graph}
	Let $P$ be a maximal ear in a digraph $D$.
	Then a set $\mathcal{I}$ of blocking intervals for $P$ can be computed in $\Bo(n + m)$ time, where $n = \Abs{\V{D}}$ and $m = \Abs{\A{D}}$.
	Further $\Abs{\mathcal{I}} \leq \Abs{\V{P}} + 2$.
	Finally, two copies of \(\mathcal{I}\) can be outputted simultaneously,
	one sorted according to the starting points of the intervals,
	and one sorted according to the end points.
\end{lemma}

\begin{proof}
	Let $\Brace{u_1, u_2, \dots, u_{k}}$ be the vertex sequence of $P$.
	Construct two arrays \(\Var{a_{\text{in}}}\) and \(\Var{a_{\text{out}}}\)
	using the vertices of \(D\) as indices and
	initializing all entries with \(\bot\).
	The value of \(\Var{a_{in}}[u]\) will be the last
	vertex along \(P\) which can reach \(u\), and
	The value of \(\Var{a_{out}}[u]\) will be the first
	vertex along \(P\) which \(u\) can reach.

	Compute the entries of \(\Var{a_{\text{out}}}\) as follows.
	Iterate through the vertices of \(D\)
	following the topological ordering of \(D\) in reverse (that is, starting with the sinks).
	When considering the vertex \(v\),
	let \(i\) be the lowest index
	such that
	\(u_i \in \OutN{v}\) or
	there is some \(w \in \OutN{v}\) with
	\(u_i = \Var{a_{\text{out}}}[w]\).
	If such an \(i\) exists, set \(\Var{a_{\text{out}}}[v] = u_i\).

	The entries of \(\Var{a_{\text{out}}}\) are computed analogously
	by starting at the sources instead of sinks, considering inneighbors instead of outneighbors, and
	taking the highest instead of the lowest index \(i\).

	Set \(\Var{d_{start}} \coloneqq \bot\) and \(\Var{d_{end}} \coloneqq \bot\).
	These variables store whether we found a deviation at the beginning or at the end of \(P\).
	Now construct an array \(\Var{b}\) using the vertices of \(P\) as indices and
	initializing all entries with \(\bot\).
	
	Iterate through the vertices of \(P\) according to their order along \(P\).
	When considering the vertex \(u_i\),
	let \(j\) be the lowest index greater than \(i+1\)
	such that
	\(u_j \in \OutN{v}\) or
	there is some \(w \in \OutN{v} \setminus \V{P}\) with
	\(u_j = \Var{a_{\text{out}}}[w]\).
	If such a \(j\) exists, set \(\Var{b}[u_j] = u_i\).

	Let \(N_{\text{in}} = \InN{u_i} \setminus \V{P}\) and
	let \(N_{\text{out}} = \OutN{u_i} \setminus \V{P}\).
	If \(\Var{d_{start}} = \bot\), \(\Abs{N_{\text{in}}} \neq \emptyset\) and \(\Var{a_{in}}[w] = \bot\) holds for all \(w \in N_{\text{in}}\),
	then set \(\Var{d_{start}} \coloneqq u_i\).
	If \(\Abs{N_{\text{out}}} \neq \emptyset\) and \(\Var{a_{out}}[w] = \bot\) holds for all \(w \in N_{\text{out}}\),
	then set \(\Var{d_{end}} \coloneqq u_i\).

	Because we only consider the neighbors of a vertex, the arrays \(\Var{a_{out}}, \Var{a_{in}}\) and \(\Var{b}\) above can be computed in \(\Bo(n + m)\) time.

	We now output the set \(\mathcal{I}\).
	To output it ordered according to the starting points of the intervals,
	first add the \(\Start{P}\)-\(\Var{d_{start}}\) subpath of \(P\) to \(\mathcal{I}\) if
	\(\Var{d_{start}}\) is not \(\bot\).
	Now iterate through the vertices of \(P\) according to their order along \(P\).

	When considering \(u_i\),
	if there is some \(w \in \OutN{u_i}\) with
	\(\Var{b}[w] = u_i\), then add the \(u_i\)-\(w\) subpath of \(P\) to \(\mathcal{I}\).

	Finally, add the \(\Var{d_{end}}\)-\(\End{P}\) subpath of \(P\) to \(\mathcal{I}\) at the end of the iteration above if
	\(\Var{d_{end}}\) is not \(\bot\).

	To output the set ordered according to the endpoints of the intervals,
	we proceed in an analogous way, but iterating through \(P\) in the reverse order.
	In particular, in both cases we only output the blocking subpath from \(u_i\) to \(u_j\)
	if \(u_i = \Var{b}[u_j]\).
	This guarantees that the same set \(\mathcal{I}\) is outputted in both iterations.
	We now show that the algorithm above is correct.

	For each vertex we add at most one blocking subpath to $\mathcal{I}$ plus at most two additional blocking subpaths for the deviations found.
	Hence, $\mathcal{I}$ contains at most $\Abs{\V{P}} + 2$ intervals.
	To complete the proof, we show that every minimal element of $\preceq$ is in $\mathcal{I}$.

	Assume towards a contradiction that there is some blocking path $B$ which is not the supergraph of any blocking subpath in $\mathcal{I}$.
	We consider two cases.

	\textbf{Case 1:} $B$ is the blocking subpath of some bypass $S$.
	Let $u_i = \Start{B}$ and $u_j = \End{B}$.
	Since \(B\) is not in \(\mathcal{I}\), we know that \(\Var{b}[u_j] \neq u_i\).
	However, since \(B\) is a path from \(u_i\) to \(u_j\), we know that \(\Var{b}[u_j] \neq \bot\).
	Hence, there is some \(u_\ell\) with \(i < \ell < j\)
	such that
	\(\Var{b}[u_j] = u_\ell\).
	However, this implies that the $u_\ell$-$u_j$ subpath $B'$ of $P$ was added to $\mathcal{I}$, and $B'$ is a subgraph of $B$, a contradiction to the assumption that \(B\) is a minimal element of \(\preceq\).

	\textbf{Case 2:} $B$ is the blocking subpath of some deviation $S$.
	Assume without loss of generality that $\End{S} \in \V{P}$.
	The case where $\Start{S} \in \V{P}$ follows analogously.
	Let $u_i = \End{S}$.
	Since $S$ is a deviation, it contains some arc $\Brace{w, u_i}$ which is not in \(P\).

	If \(\Var{d_{start}} \neq u_i\),
	then it must have been set to some \(u_\ell\) where \(\ell < i\).
	However, this implies that the $\Start{P}$-$u_\ell$ subpath $B'$ of $P$ was added to $\mathcal{I}$, and $B'$ is a subgraph of $B$, a contradiction to the assumption that \(B\) is a minimal element of \(\preceq\).

	Since both \textbf{Case 1} and \textbf{Case 2} lead to a contradiction, we conclude that every minimal element of $\preceq$ is in $\mathcal{I}$, as desired.
\end{proof}

Combining the previous results, we can now conclude that \PathPA/ is in \P/ if the input digraph is acyclic.

\begin{theorem}
		\label{state:local-ear-anonymity-in-P-on-DAGs}
		Given an acyclic digraph $D$ and a maximal ear $P$ in $D$, we can compute $\Anon[D](P)$ and find an identifying sequence for $P$ of minimum length in $\Bo(n + m)$ time, where $n = \Abs{\V{D}}$ and $m = \Abs{\A{D}}$.
	\end{theorem}
	
	\begin{proof}
		Compute a set $\mathcal{I}$ of blocking intervals for $P$ in $\Bo(n + m)$-time using \cref{state:compute-blocking-interval-graph}.
		Note that $\mathcal{I}$ has at most $\Abs{\V{P}} + 2 \in \Bo(n)$ intervals.
		Further, we obtain two copies of \(\mathcal{I}\),
		one sorted according to the starting points of the intervals, and
		one sorted according to the endpoints.
		Then, compute a minimum hitting set $\bar{a}$ for $\mathcal{I}$ in $\Bo(n)$ time using \cref{state:compute-interval-hitting-set}.
		By \cref{state:hitting-set-implies-identifying-sequence}, $\bar{a}$ is an identifying sequence for $P$.
		Further, by \cref{state:identifying-sequence-implies-hitting-set} every identifying sequence for $P$ must be at least as long as $\bar{a}$.
		Hence, $\Anon[D](P) = \Abs{\bar{a}}$.
		Finally, the overall running time is in $\Bo(n + m)$.
	\end{proof}
	
	Further, using similar methods as in \cref{state:local-ear-anonymity-in-P-on-DAGs}, we can also solve \EarIS/ in polynomial time if the input digraph is acyclic.

\begin{theorem}
		\label{state:eis-in-P-on-dags}
		\EarIS/ can be solved in $\Bo(n+m)$ time if the input digraph $D$ is acyclic, where $n = \Abs{\V{D}}$ and $m = \Abs{\A{D}}$.
	\end{theorem}
	
	\begin{proof}
		First, sort the vertices of \(D\) according to their topological ordering.
		Then, compute the blocking interval set $\mathcal{I}$ for $P$ in $\Bo(n+m)$ time using \cref{state:compute-blocking-interval-graph}.
		By \cref{state:identifying-sequence-implies-hitting-set,state:hitting-set-implies-identifying-sequence}, $\bar{a}$ is an identifying sequence for $P$ if, and only if, it is a hitting set for $\mathcal{I}$.

		We can verify if $\bar{a}$ is a hitting set for $\mathcal{I}$ by iterating through \(\mathcal{I}\) and
		\(\bar{a}\) simultaneously as follows.

		First, we use the copy of \(\mathcal{I}\) which is sorted according to the endpoints of the intervals.
		If the arc \(a_i\) being considered lies between the indices of the interval \(I_j\), we hit \(I_j\) and can move to the next interval.
		Otherwise, \(a_i\) does not hit any further intervals and we move to the next arc along \(\bar{a}\).
		If we exhaust \(\bar{a}\) before hitting all elements of \(\mathcal{I}\), then \(\bar{a}\) is not a hitting set for \(\mathcal{I}\).
		Otherwise, it is a hitting set for \(\mathcal{I}\).
		This step can be done in linear time, as the indices of vertices are integers and both \(\mathcal{I}\) and \(\bar{a}\) are already sorted.
		Hence, the overall running time is in $\Bo(n + m)$, as desired.
	\end{proof}

	In order to solve \PAlong/ in polynomial time on DAGs, we compute for each vertex $v$ a number $\Var{anon}[v]$ which is a lower bound
	to the number of arcs required in the ear-identifying sequence of any maximal path containing \(v\).
	We do this by following the topological ordering of the vertices and
	by keeping track of the bypasses and deviations found.

	The algorithm relies on the property of DAGs that, given four distinct vertices $v_1, v_2, v_3, v_4$,
	sorted according to their topological ordering,
	if $v_2$ can reach $v_3$,
	then every $v_1$-$v_2$ path is disjoint from every $v_3$-$v_4$ path.
	This allows us to efficiently compute bypasses using breadth-first search.
	The pseudo-code is provided in \cref{alg:dag-ear-anonymity}.

\begin{varalgorithm}{DAGEarAnonymity}
	\begin{algorithmic}[1]
		\Function{DAGEarAnonymity}{DAG $D$}
		\State $V \gets $ sort $\V{D}$ by the topological ordering of $D$
		\State $\Var{anon} \gets $ empty array over $V$
		\ForAll{$v \in V$}
			\label{alg:dag-ear:outer-for-loop}
			\State $\Var{anon}[v] \gets \max(\{\Var{anon}[u] \ST u \in \InN{v} \} \cup \{0\}$)
			\label{alg:dag-ear:inneighbor-set}
			\If{$\In{v} > 1$ and $\Var{anon}[v] = 0$}
				\State $\Var{anon}[v] \gets 1$
				\label{alg:dag-ear:first-deviation-set}
			\Else
				\ForAll{$u \in \InN{v}$}
					\State $V_u \gets $ vertices which can reach $u$
					\label{alg:dag-ear:V_u}
					\State \( U_u \gets 
						\{ w \in V_u
						\ST
						\text{ there is a \(w\)-\(v\) path in \(D - u\)}\)
					\label{alg:dag-ear:U_u}
					\State \hspace{1.25cm}
						which is internally disjoint from $V_u$
						\}
					\State $\Var{anon}[v] \gets \max(\{\Var{anon}[w] + 1 \ST w \in U_u\} \cup \Set{\Var{anon}[v]})$
					\label{alg:dag-ear:bypass-set}
				\EndFor
				\If{$\Out{v} = 0$ and \(\In{v} > 0\)}
					\State \(P \gets \) shortest path ending on $v$ such that $\In{\Start{P}} > 1$
					\label{alg:dag-ear:sink-path-P}
					\If{no such $P$ exists}
						\State \(P' \gets \) the unique maximal path with \(\End{P'} = v\)
						\label{alg:dag-ear:no-deviation-set}
											\ElseIf{there is some \(u \in \V{P}\) with \(\Out{u} > 1\)}
						\State \(P' \gets \) shortest subpath of $P$ ending on $v$ such that 
						\label{alg:dag-ear:no-incoming-deviation-set}
						\State \hspace{1.0cm} $\Out{\Start{P'}} > 1$
					\Else
						\State \textbf{continue}
					\EndIf
					\ForAll{$u \in \V{P'} \setminus \{\Start{P'}\}$}
						\State $\Var{anon}[u] \gets \Var{anon}[u] + 1$
						\label{alg:dag-ear:last-deviation-set}
					\EndFor
				\EndIf
			\EndIf
		\EndFor
		\State \Return \(\max(\{\Var{anon}[v] \ST v \in \V{D}\})\)
		\EndFunction
		\end{algorithmic}
		\caption{Compute the ear anonymity of a DAG.}
		\label{alg:dag-ear-anonymity}
\end{varalgorithm}

\begin{lemma}
	\label{state:dag-ear-anonymity:at-least}
	At the end of the execution of \cref{alg:dag-ear-anonymity},
	for every \(v \in \V{D}\) there is 
	a path \(P\) which 
	starts at some source of \(D\) and
	ends in \(v\)
	such that,
	for every path \(Q\) starting in \(v\) and
	ending in some sink,
	the path \(R \coloneqq P \cdot Q\) has \(\Anon[D](R) \geq \Var{anon}[v]\).
	Furthermore, there is an ear-identifying sequence \(\bar{a}\) of minimum length for \(R\)
	such that
	at least \(\Var{anon}[v]\) arcs of \(\bar{a}\) lie in \(P\).
\end{lemma}
\begin{proof}
	We prove the statement by induction on the index of \(v\)
	in the topological ordering of \(D\).

	The statement is clearly true if \(v\) is a source.
	So assume that \(v\) is not a source.
	
	\begin{CaseDistinction}
		\Case{The last change in the value of \(\Var{anon}[v]\) was on line \ref{alg:dag-ear:inneighbor-set}.}
		Then there is some \(u \in \InN{v}\)
		such that
		\(\Var{anon}[u] = \Var{anon}[v]\).

		By the induction hypothesis,
		there is some \(P\) which starts in a source and ends in \(u\)
		satisfying the additional conditions given at the statement.
		Let \(Q\) be some path starting in \(v\) and ending in some sink.
		Let \(P' = P \cdot (u, v)\).

		By assumption, there is some ear-identifying sequence \(\bar{a}\) for \(R \coloneqq P \cdot (u,v) \cdot Q\)
		such that
		at least \(\Var{anon}[u] = \Var{anon}[v]\) arcs of \(\bar{a}\) lie in \(P\) and, hence, in \(P'\).

		\Case{The last change in the value of \(\Var{anon}[v]\) was on line \ref{alg:dag-ear:first-deviation-set}.}
		Then \(\Var{anon}[u] = 0\) for all inneighbors \(u \in \InN{v}\) of \(v\).
		Let \(u_1, u_2\) be two distinct inneighbors of \(v\).
		By the induction hypothesis, there is some path \(P\) starting at some source and
		ending in \(u_1\) satisfying the properties given in the statement.

		Let \(P' = P \cdot (u_1, v)\).
		Let \(Q\) be some path starting in \(v\) and ending in some sink.
				Let \(R' = P' \cdot Q\) and
		let \(\bar{a}\) be some ear-identifying sequence for \(R'\) of minimum length.
		As there are at least two maximal paths (one coming from \(u_1\) and the other from \(u_2\)) visiting all arcs of \(\bar{a}\) lying on \(Q\), some arc of \(P'\) must be on \(\bar{a}\).
								Hence, \(P'\) contains at least \(\Var{anon}[v] = 1\) arcs of \(\bar{a}\), as desired.

		\Case{The last change in the value of \(\Var{anon}[v]\) was on line \ref{alg:dag-ear:bypass-set}.}
		Then there are vertices \(w, u_1\)
		such that
		\(u_1\) is an inneighbor of \(v\),
		\(\Var{anon}[w] = \Var{anon}[v] - 1\),
		there is a \(w-v\) path \(P_1\) which is disjoint from \(u_1\), and
		\(w\) can reach \(u_1\).
				By the induction hypothesis, there is a path \(P\) starting at a source and
		ending in \(w\) which satisfies the conditions given in the statement.

		Let \(P' = P \cdot P_1\)
		Let \(Q\) be some path starting at \(v\) and ending at some sink.
		By assumption, there is an ear-identifying sequence \(\bar{a}\) for \(R \coloneqq P \cdot P_1 \cdot Q\) of minimum length
		such that
		at least \(\Var{anon}[w]\) arcs of \(\bar{a}\) lie on \(P\) and hence on \(P'\).
		However, there are at least two different paths from \(w\) to \(v\) (one visiting \(u_1\) and one not).
		Hence, at least one arc of \(\bar{a}\) must lie on \(P_1\) and hence on \(P'\).

		\Case{The last change in the value of \(\Var{anon}[v]\) was on line \ref{alg:dag-ear:last-deviation-set}.}

		If the path \(P'\) was chosen on line \ref{alg:dag-ear:no-deviation-set},
		then there is exactly one maximal path \(Q\) containing \(v\).
		In particular, every vertex in \(Q\) has indegree at most one.
		Hence, \(\Anon[D](Q) = \Var{anon}[v] = 1\) and
		the statement follows.

		Otherwise, there is a path \(P\) with
		\(\In{\Start{P}} > 1\) which
		ends on some sink \(t\) and
		contains \(v\).
		Furthermore, there is exactly one \(v\)-\(t\) path in \(D\).
		Let \(P'\) be the shortest subpath of \(P\) ending on \(t\)
		such that
		\(\Out{\Start{P'}} > 1\).
		Note that \(P'\) contains \(v\) since
		\(\Var{anon}[v]\) was incremented on line \ref{alg:dag-ear:last-deviation-set}.
		Further, \(\In{u} = 1\) holds for all \(u \in V(P') \setminus \{\Start{P'}\}\).
		In particular, \(\In{v} = 1\)
		and the value of \(\Var{anon}[v]\) was not modified on line
		\ref{alg:dag-ear:first-deviation-set} nor on
		\ref{alg:dag-ear:bypass-set}.

		Let \(u \in \InN{v}\).
		By the induction hypothesis, there is a path
		\(P_u\) ending on \(u\) and
		satisfying the conditions in the statement.
		Note that \(\Var{anon}[v] \leq \Var{anon}[u] + 1\).
		
		If \(u \in \V{P'} \setminus \{\Start{P'}\}\), then \(\Var{anon}[u] = \Var{anon}[v]\) and
		there is exactly one \(u\)-\(v\) path in \(D\).
		The statement follows trivially.

		Otherwise we have \(u = \Start{P'}, \Out{u} > 1\) and \(\In{v} = 1\).
		Further, there is exactly one path \(Q\) starting in \(v\) and ending on some sink, and
		there is exactly one \(u\)-\(v\) path.
		Let \(R = P_u \cdot (u,v) \cdot Q\) and
		let \(w \in \Out{u} \setminus \{v\}\).

		By the induction hypothesis,
		there is an ear-identifying sequence \(\bar{a}\) of minimum length for \(R\) 
		such that
		at least \(\Var{anon}[u]\) arcs of \(\bar{a}\) lie in \(P_u\).
		Since \((u,w)\) is a deviation for \(R\),
		at least one arc of \(\bar{a}\)
		must lie on \((u,v) \cdot Q\).
		As \(Q\) is the only maximal path starting in \(v\),
		we choose \(\bar{a}\)
		such that
		it contains the arc \((u,v)\) and
		no arcs in \(Q\).
		Hence, at least \(\Var{anon}[v] = \Var{anon}[u] + 1\) arcs of \(\bar{a}\)
		lie in \(P_u \cdot (u, v)\), as desired.
	\end{CaseDistinction}
\end{proof}

\begin{lemma}
	\label{state:dag-ear-anonymity:at-most}
	At the end of the execution of \cref{alg:dag-ear-anonymity},
	for every \(v \in \V{D}\) and
	every maximal path \(R \coloneqq P \cdot v \cdot Q\)
	there is an ear-identifying sequence \(\bar{a}\) of minimum length for \(R\) 
	such that
	at most \(\Var{anon}[v]\) arcs of \(\bar{a}\) lie in \(P\).
\end{lemma}
\begin{proof}
	We prove the following slightly stronger statement.

	\textbf{Claim:}
	for every \(v \in \V{D}\) and
	every maximal path \(R \coloneqq P \cdot v \cdot Q\)
	there is an ear-identifying sequence \(\bar{a}\) of minimum length for \(R\) 
	such that
	for every \(u \in V(P \cdot v)\),
	at most \(\Var{anon}[u]\) arcs of \(\bar{a}\) lie in \(P_u\),
	where \(P_u\) is the subpath of \(P\)
	starting on \(\Start{P}\) and ending on
	\(u\).

	We prove the statement by induction on the index of \(v\)
	in the topological ordering of \(D\).

	The statement is clearly true if \(v\) is a source, as \(P\) then becomes empty.
	So assume that \(v\) is not a source and
	let \(R \coloneqq P \cdot (u, v) \cdot Q\) be a maximal path,
	where \(u\) is the predecessor of \(v\) along \(R\). 
	Let \(P_v = P \cdot (u, v)\). 

	By the induction hypothesis,
	there is an ear-identifying sequence \(\bar{a}\)
	of minimum length for \(R\)
	which satisfies the condition given in the claim above.
	In particular,
	at most \(\Var{anon}[u]\) arcs of \(\bar{a}\) lie in \(P\).

	If \(\Var{anon}[v] > \Var{anon}[u]\), then clearly
	at most \(\Var{anon}[u] + 1 \leq \Var{anon}[v]\) arcs of \(\bar{a}\) lie in \(P_v\).

	Now assume that \(\Var{anon}[v] = \Var{anon}[u]\).
	If \(\Var{anon}[v] = 0\), then \(\In{u'} \leq 1\) holds for all \(u'\) which can reach \(v\),
	as otherwise we would increment \(\Var{anon}[u']\) on line \ref{alg:dag-ear:first-deviation-set} and
	propagate this through line \ref{alg:dag-ear:inneighbor-set}.
	In particular, there is exactly one path which starts at some source and ends in \(v\).
	Since \(v\) is not a source and
	\(\Var{anon}[v]\) was not incremented on line \ref{alg:dag-ear:last-deviation-set},
	we know that \(v\) can reach some vertex \(w'\) with \(\Var{anon}[w'] \geq 1\).
	This means that \(v\) has at least one outneighbor \(w\).
	Hence, any sequence containing some arc in \(Q\) satisfies the required condition.
	Thus, we can assume that \(\Var{anon}[v] \geq 1\).
				
	If \(\bar{a}\) does not contain \((u, v)\), then there is nothing to show.
	Further, if less than \(\Var{anon}[u]\) arcs of \(\bar{a}\) lie in \(P\),
	then clearly at most \(\Var{anon}[u] = \Var{anon}[v]\) arcs of \(\bar{a}\)
	lie in \(P_v\).

	Let \(P_a\) be the shortest subpath of \(P\) starting at \(\Start{P}\)
	which contains the same arcs of \(\bar{a}\) as \(P\),
	and let \(P_b\) be the rest of \(P\), that is, \(P = P_a \cdot P_b\).
	Observe that, by the induction hypothesis and
	by the assignment on line \ref{alg:dag-ear:inneighbor-set},
	\(\Var{anon}[u'] = \Var{anon}[u]\) holds for all \(u' \in \V{P_b}\).

	If there is some path \(R'\) starting in \(V(P_b)\) and
	ending on \(v\) without using \((u,v)\), then \(\Var{anon}[v]\)
	is incremented on line \ref{alg:dag-ear:bypass-set}, as either \(\Start{R'}\)
	can reach \(u\) or \(u\) can reach the predecessor of \(v\) on \(R'\).
	This contradicts, however, the equality \(\Var{anon}[v] = \Var{anon}[u]\) assumed previously.

	Assume towards a contradiction that,
	if \(\bar{b}\) is an ear-identifying sequence of minimum length for \(R\)
	satisfying the conditions in the claim above, then
	\(\bar{b}\) contains \((u, v)\) and
	exactly \(\Var{anon}[u]\) arcs of \(\bar{b}\) lie in \(P\).
	We consider the following cases.

	\begin{CaseDistinction}
		\Case{\(v\) is not a sink.}
	
		Let \(w\) be the successor of \(v\) in \(R\) and
		let \(\bar{b}_1\) be the arc sequence obtained by
		replacing \((u,v)\) with \((v, w)\) in \(\bar{b}\)
		while preserving the topological ordering of the arcs.

		By assumption, there is a conflicting ear \(R'\) for \((R, \bar{b}_1)\).
		Further, \(R'\) does not contain \((u, v)\)
		as \(D\) is acyclic and \(\bar{b}\) is an ear-identifying sequence for \(R\).

		If \(R'\) is disjoint from \(P\),
		then \(\Var{anon}[u] = 0\), a contradiction.

		Otherwise, 
		\(\Start{R'}\) must lie on \(P_b\), as \(R'\) must visit all arcs of \(\bar{a}\) which are in \(P_a\).
		Since \(R'\) contains \(v\), it also contains a path from \(\V{P_b}\) to \(v\) avoiding \((u,v)\).
		This contradicts the argumentation above before the case distinction.

		\Case{\(v\) is a sink.}

		Let \(\bar{b}_1\) be the subsequence of \(\bar{b}\) obtained
		by removing \((u,v)\) from \(\bar{b}\).
		By assumption, there is a conflicting ear \(R'\) for
		\((R, \bar{b}_1)\).
		In particular, \(R'\) intersects \(P_a\), as at least \(\Var{anon}[u] \geq 1\) arcs of
		\(\bar{a}\) lie in \(P_a\).
		This implies that \(R'\) contains a path from \(\V{P_b}\) to \(v\) which avoids \((u,v)\),
		a contradiction.
	\end{CaseDistinction}
\end{proof}

	\Cref{state:dag-ear-anonymity:at-least,state:dag-ear-anonymity:at-most} essentially prove that
	\cref{alg:dag-ear-anonymity} is correct.
	Hence, we now only need to provide a running-time analysis.

	\begin{theorem}
		\label{state:ear-anonymity-in-p-on-dags}
		\PAlong/ can be solved in $\Bo(m(n + m))$-time if the input digraph $D$ is acyclic,
		where $n = \Abs{\V{D}}$ and $m = \Abs{\A{D}}$.
	\end{theorem}
	\begin{proof}
		We run \cref{alg:dag-ear-anonymity} on the input digraph, obtaining a value \(k = \Var{anon}[v]\)
		for some sink \(v \in \V{D}\) for which \(\Var{anon}[v]\) is maximum.
		By \cref{state:dag-ear-anonymity:at-least}, \(k \leq \Anon(D)\).
		By \cref{state:dag-ear-anonymity:at-most}, \(k \geq \Anon(D)\), and so \(k = \Anon(D)\).

		We now analyze the running time of \cref{alg:dag-ear-anonymity}.
		Sorting \(\V{D}\) according to the topological ordering can be done in 
		\(\Bo(n + m)\)-time using standard techniques.

		The \textbf{for}-loop on line \ref{alg:dag-ear:outer-for-loop} is executed exactly \(n\) times.
		We then iterate over all inneighbors of \(v\).
		Hence, each arc is considered a constant number of times.
		For each arc, we compute the sets \(V_u, U_u\) on lines \ref{alg:dag-ear:V_u} and \ref{alg:dag-ear:U_u}
		using breadth-first searches in \(\Bo(n + m)\)-time. 
		For each sink, the paths \(P, P'\) on lines
		\ref{alg:dag-ear:sink-path-P}, \ref{alg:dag-ear:no-deviation-set} and \ref{alg:dag-ear:no-incoming-deviation-set}
		can also be computed with breadth-first search.
		Hence, the running time is dominated by computing a constant number of
		breadth-first searches for each arc, 
		and so it lies in \(\Bo(m(n + m))\).
	\end{proof}

\section{\NP/-hardness}
\label{sec:eis-is-np-hard}

	We consider the problems \EarIS/, \PathPA/ and \PAlong/ in the general setting without any restrictions on the input digraph.
	We show that \EarIS/ is \NP/-hard, providing a reduction from the \NP/-complete problem \kLinkage/, defined below.

	\kLinkageDef

	\kLinkage/ remains \NP/-hard even if $k = 2$. \cite{Fortune:1978}

	\begin{theorem}
		\label{state:ear-id-sequence-is-NP-hard}
		\EarIS/ is \NP/-complete even if $\bar{a}$ has length 3.
	\end{theorem}
	\begin{proof}
		From \cref{state:eis-is-in-np} we know that \EarIS/ is in \NP/.
		To show that it is \NP/-hard, we provide a reduction as follows.
				Let $\Brace{D, S}$ be a \kLinkage/ instance where $k = \Abs{S} = 2$.
		Construct a digraph $D'$ as follows (see \cref{fig:eis-hardness-reduction} below for an illustration of the construction).

		Start with $D$.
		Add the vertices $\Set{u_1, \ldots, u_3, v_1, \ldots, v_6}$ and the following paths to $D'$, where each path is given by its vertex-sequence:
		\begin{align*}
			P_1 & = \Brace{u_1, v_1, v_2, u_2, v_3, v_4, u_3, v_5, v_6}
			,\\ P_2 & = \Brace{v_2, u_1, s_1}, P_3 = \Brace{t_1, u_3, v_3}, P_4 = \Brace{v_4, u_2, s_2}, P_5 = \Brace{t_2, v_5}.
		\end{align*}
		Set $\bar{a} = \Brace{\Brace{v_1, v_2}, \Brace{v_3,v_4}, \Brace{v_5, v_6}}$.
		Note that $\bar{a}$ is a sequence of arcs of $P_1$, sorted according to their occurrence on $P_1$.
		This completes the construction of the \EarIS/ instance $\Brace{D', P_1, \bar{a}}$.
		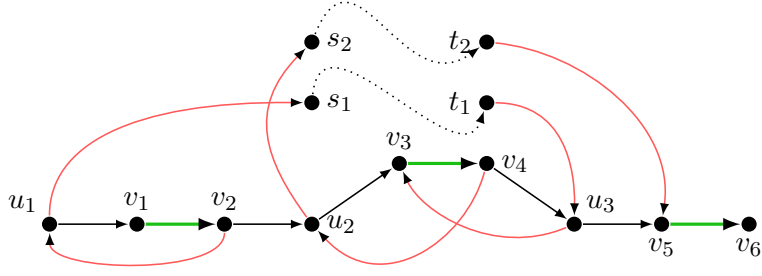
\begin{figure}[H]
		\centering
		\vspace{-1.0cm}
			\begin{tikzpicture}[yscale=0.7]
			\node[circle, label = above left:{{$u_1$}}, line width = 0.896, fill = black]
	(v1) at (0, 0){};
\node[circle, label = above:{{$v_1$}}, line width = 0.896, fill = black]
	(v2) at (1.15, 0){};
\node[circle, label = above:{{$v_2$}}, line width = 0.896, fill = black]
	(v3) at (2.3, 0){};
\node[circle, label = right:{{$u_2$}}, line width = 0.896, fill = black]
	(v4) at (3.45, 0){};
\node[circle, label = above right:{{$u_3$}}, line width = 0.896, fill = black]
	(v6) at (6.9, 0){};
\node[circle, label = above:{{$v_3$}}, line width = 0.896, fill = black]
	(v11) at (4.6, 1.15){};
\node[circle, label = right:{{$v_4$}}, line width = 0.896, fill = black]
	(v12) at (5.75, 1.15){};
\node[circle, label = right:{{$s_1$}}, line width = 0.896, fill = black]
	(v14) at (3.45, 2.3){};
\node[circle, label = left:{{$t_1$}}, line width = 0.896, fill = black]
	(v16) at (5.75, 2.3){};
\node[circle, label = right:{{$s_2$}}, line width = 0.896, fill = black]
	(v18) at (3.45, 3.45){};
\node[circle, label = left:{{$t_2$}}, line width = 0.896, fill = black]
	(v20) at (5.75, 3.45){};
\node[circle, label = below:{{$v_5$}}, line width = 0.896, fill = black]
	(v23) at (8.05, 0){};
\node[circle, label = below:{{$v_6$}}, line width = 0.896, fill = black]
	(v24) at (9.2, 0){};
\path[-latex, line width = 0.592, draw = black]
	(v1) to (v2);
\path[-latex, line width = 1.2, draw = green]
	(v2) to (v3);
\path[-latex, line width = 0.592, draw = black]
	(v3) to (v4);
\path[-latex, line width = 1.2, draw = green]
	(v11) to (v12);
\path[-latex, line width = 0.592, draw = black]
	(v12) to (v6);
\path[-latex, line width = 0.592, draw = black]
	(v6) to (v23);
\path[-latex, line width = 1.2, draw = green]
	(v23) to (v24);
\path[-latex, line width = 0.592, draw = red]
	(v3) .. controls (2.305, -0.987) and (0.01, -0.917) .. (v1);
\path[-latex, line width = 0.592, draw = red]
	(v1) .. controls (-0.006, 1.486) and (1.065, 2.333) .. (v14);
\path[-latex, line width = 0.592, draw = red]
	(v16) .. controls (6.783, 2.293) and (6.912, 1.03) .. (v6);
\path[-latex, line width = 0.592, draw = red]
	(v6) .. controls (5.836, -0.642) and (4.888, 0.242) .. (v11);
\path[-latex, line width = 0.592, draw = red]
	(v12) .. controls (5.471, -0.34) and (4.358, -1.419) .. (v4);
\path[-latex, line width = 0.592, draw = red]
	(v4) .. controls (2.844, 1.191) and (2.58, 2.177) .. (v18);
\path[-latex, line width = 0.592, draw = red]
	(v20) .. controls (7.039, 3.265) and (8.232, 1.87) .. (v23);
\path[-latex, line width = 0.592, draw = black, dotted]
	(v14) .. controls (3.582, 4.188) and (5.402, 0.813) .. (v16);
\path[-latex, line width = 0.592, draw = black, dotted]
	(v18) .. controls (4.022, 5.474) and (4.65, 2.286) .. (v20);
\path[-latex, line width = 0.592, draw = black]
	(v4) to (v11);

		\end{tikzpicture}
		\caption{Gadget for the reduction in the proof of \cref{state:ear-id-sequence-is-NP-hard}.
		Bold, green arcs belong to $\bar{a}$, and the red arcs belong to the conflicting ear $Q$ but not to $P_1$.
	The dotted lines correspond to disjoint paths in $D$.}
		\label{fig:eis-hardness-reduction}
	\end{figure}

		We first show that, if the \kLinkage/ instance $\Brace{D, S}$ is a ``yes''-instance, then so is the $\EarIS/$ instance $\Brace{D', P_1, \bar{a}}$.

		Let $\mathcal{L}$ be a solution for $\Brace{D, S}$.
		Let $L_i$ be the $s_i$-$t_i$ path in $\mathcal{L}$ for $i \in \Set{1,2}$.
		We construct a conflicting ear $Q$ for $\Brace{P_1, \bar{a}}$ as follows.
		We set $Q = \Brace{v_1, v_2} \cdot P_2 \cdot L_1 \cdot P_3 \cdot \Brace{v_3, v_4} \cdot P_4 \cdot L_2 \cdot P_5 \cdot \Brace{v_5, v_6}$.

		Clearly $Q$ is not a subgraph of $P_1$ and $Q$ visits the arcs of $\bar{a}$ in the given order.
		By assumption, $L_1$ and $L_2$ are disjoint paths, and so $Q$ is a path.
		Hence, $Q$ is a conflicting ear for $\Brace{P_1, \bar{a}}$, as desired.

		For the other direction, let $Q$ be a conflicting ear for $\Brace{P_1, \bar{a}}$.
		We first show that $Q$ does not contain $\Brace{v_2, u_2}$.
		Assume towards a contradiction that it does contain $\Brace{v_2, u_2}$.
		Then $Q$ does not contain $\Brace{v_2, u_1}$ or $\Brace{v_4, u_2}$.
		Since $Q$ contains $\Brace{v_3, v_4}$, it must also contain $\Brace{v_4, u_3}$.
		However, $\Brace{u_3, v_3}$ closes a cycle with the arcs $\Brace{v_3, v_4}$ and $\Brace{v_4, u_3}$.
		Hence, $Q$ cannot contain $\Brace{u_3, v_3}$.
		This implies that $Q$ must contain $\Brace{u_2, v_3}$ in order to reach $v_3$.
		Finally, $Q$ must contain $\Brace{u_3, v_5}$ and $\Brace{v_5, v_6}$.
		However, we now have $Q = P_1$, a contradiction to the assumption that $Q$ is a conflicting ear for $\Brace{P_1, \bar{a}}$.

		Since $Q$ contains $\Brace{v_1, v_2}$ but not $\Brace{v_2, u_2}$, it must contain $\Brace{v_2,u_1}$.
		As $\Brace{u_1, v_1}$ closes a cycle, $Q$ does not contain this arc and must contain $\Brace{u_1, s_1}$ instead.
		Because $Q$ contains $\Brace{v_3, v_4}$, it must reach $v_3$ through $u_3$ or through $u_2$.
		However, if $Q$ contains $\Brace{u_2, v_3}$, then it also contains $\Brace{v_4, u_2}$, which closes a cycle.
		Hence, $Q$ does not contain $\Brace{u_2, v_3}$ and must contain $\Brace{u_3, v_3}$ instead.
		As before, $Q$ cannot contain $\Brace{v_4, u_3}$ as this would close a cycle, so $Q$ must contain $\Brace{v_4, u_2}$ and $\Brace{u_2, s_2}$.

		In order to reach $u_3$, $Q$ must contain $\Brace{t_1, u_3}$.
		Since $Q$ contains both $s_1$ and $s_2$, it must also contain $t_1$ and $t_2$,
		as they are the only vertices of $D$ which are reachable by $s_1$ and $s_2$ and have arcs to $P_1$.
		Hence, $Q$ contains $\Brace{t_2, v_5}$.
		Since $v_5$ can only reach $v_6$, $Q$ must visit $(t_1,u_3)$ before visiting $(t_2, v_5)$.
		Further, $Q$ visits $t_1$ before visiting $s_2$ and it visits $s_1$ before visiting $t_1$.
		Hence, $Q$ must visit \(t_2\) after \(s_2\) and must also contain two paths $L_1$ and $L_2$, where $L_1$ is an $s_1$-$t_1$ path in $D$ and $L_2$ is an $s_2$-$t_2$ path in $D$.
		Since $Q$ is an ear, $L_1$ and $L_2$ must be disjoint.
		Thus, $L_1$ and $L_2$ are a solution to the \kLinkage/ instance $\Brace{D, S}$, as desired.
	\end{proof}

	Using \cref{state:ear-id-sequence-is-NP-hard,state:identifying-sequence-implies-hitting-set}, it is simple to show that \PathPA/ is \NP/-hard as well.

		\begin{theorem}
			\label{state:ear-ear-anonymity-is-np-hard}
		\PathPA/ is \NP/-hard even if $k \leq 3$.
	\end{theorem}
	
	\begin{proof}
		We provide a reduction from \EarIS/, which is \NP/-hard even if the input sequence has length 3.
		Let $\Brace{D, P, \bar{a}}$ be an \EarIS/ instance where $\Abs{\bar{a}} = 3$.
		Let $\bar{a} = \Brace{\Brace{v_1, v_2}, \Brace{v_3, v_4}, \Brace{v_5, v_6}}$.
		We construct an \PathPA/ instance $\Brace{D', P', k}$ as follows.

		Set $k \Setv 3$ and $P' \Setv P$.
		Construct $D'$ by first adding the vertices $u_1, u_2$ and $u_3$ to $D$.
		Then, add the paths $Q_1, Q_2$ and $Q_3$, where $Q_1 = \Brace{v_1, u_1, v_2}$, $Q_2 = \Brace{v_3, u_2, v_4}$ and $Q_3 = \Brace{v_5, u_3, v_6}$, given by their vertex-sequences.
		This completes the construction.

		The paths $Q_1, Q_2$ and $Q_3$ are each a bypass for $P'$.
		Furthermore, $\Brace{v_{1}, v_{2}}, \Brace{v_{3}, v_{4}}, \Brace{v_{5}, v_{6}}$ are the blocking subpaths for $Q_1, Q_2, Q_3$, respectively.
		By \cref{state:identifying-sequence-implies-hitting-set}, every identifying sequence for $P'$ must contain the arcs of $\bar{a}$.

		As the arcs in $\bar{a}$ contain all starting and endpoints of $Q_1, Q_2$ and $Q_3$, no conflicting ear for $\Brace{P', \bar{a}}$ in $D'$ can contain $Q_1, Q_2$ or $Q_3$.
		Thus, an ear $Q$ in $D'$ is a conflicting ear for $\Brace{P', \bar{a}}$ if, and only if, $Q$ is a conflicting ear for $\Brace{P, \bar{a}}$.
		Hence, there is an identifying sequence for $P'$ of length at most 3 in $D'$ if, and only if, $\bar{a}$ is an identifying sequence for $P$ in $D$, as desired.
	\end{proof}

	For the next hardness result, we provide a reduction from the following \NP/-complete problem. \cite{Fortune:1978}

\UVWPathDef

\begin{theorem}
	\label{thm:path anonymity is NP-hard}
	\PAlong/ is \NP/-hard.
\end{theorem}

\begin{proof}
	Let $(D, u, v, w)$ be a $\UVWPath/$ instance.
	Let $n = |V(D)|$.
	We assume, without loss of generality, that $\In{u} = 0$ and $\Out{w} = 0$, as no $u$-$w$ path can contain incoming arcs of $u$ or outgoing arcs of $w$.
	We construct an $\PAlong/$ instance $(D', k)$ as follows.

	Start by setting $D' \Setv D$.
	Then add the digraphs $H_1, H_2$ and $H_3$ to $D'$, where $H_i$ is defined as (see \cref{fig:np-hard Hi} for an illustration of $H_1$)
	\begin{align*}
		V(H_i) & = \{ u_{i,0}, u_{i,1}, \dots u_{i,4n} \},\\
		A(H_i) & = \{ (u_{i,cj}, u_{i,c(j + 1)}) \mid c \in \{1,2\} \text{ and } 0 \leq j \leq 4n/c - 1 \}.
	\end{align*}
	Add the arcs $(u_{1,4n}, u)$ and $(w, u_{3,0})$.
	Remove the vertex $v$.
	Add the vertices $v_1, v_2$, together with the arcs $\InN[D]{v} \times \{v_1\}, \{v_2\} \times \OutN[D]{v}$, $(v_1, u_{2,0})$ and $(u_{2,4n}, v_2)$.
	Now set $k \Setv 6n$.

	We show that there is an $u$-$w$ path in $D$ going through $v$ if and only if $\Anon(D') \geq k$.

	Assume there is an $u$-$w$ path $P$ visiting $v$ in $D$.
	Partition $P$ as $P = P_1 \cdot v \cdot P_2$.
	For each $1 \leq i \leq 3$, let $Q_i$ denote the path $\Brace{u_{i,0}, u_{i,2}, \ldots u_{i,4n}}$ given by its vertex sequence.
	Set $Q \Setv Q_1 \cdot u \cdot P_1 \cdot v_1 \cdot Q_2 \cdot v_2 \cdot P_2 \cdot w \cdot Q_3$.
	Note that $Q$ is maximal as $Q_1$ starts in a source and $Q_3$ ends in a sink.
	Since every arc of $Q_1$, $Q_2$ and $Q_3$ is the blocking subpath of some bypass for $Q_1, Q_2$ or $Q_3$, any minimum hitting set of any set of hitting intervals for $Q$ must contain at least $6n$ arcs.
	Hence, from \cref{state:identifying-sequence-implies-hitting-set} we obtain that $\Anon[D](Q) \geq 6n = k$ and so $\Anon(D) \geq k$, as desired.

	Now assume there is some maximal path $Q$ in $D'$ with $\Anon[D](Q) \geq k$.
	The path $Q$ cannot be completely contained inside $D$, since $\Anon[D](Q) \leq \Abs{\V{Q}}$ and $\Abs{\V{D}} = n < k = 6n$.
	If $Q$ contains some vertex of $H_1$, then it necessarily contains both $u$ and $u_{1,0}$, as otherwise $Q$ would not be maximal.
	Similarly, $Q$ contains $v_1, v_2, u_{2,0}$ and $u_{2,4n}$ if it contains some vertex of $H_2$, and $Q$ contains $w$ and $u_{3,4n}$ if it contains some vertex of $H_3$.

	If $Q$ contains vertices of $H_1, H_2$ and $H_3$, we claim that there is a $u$-$w$ path in $D$ containing $v$.
	As $u_{2,0}$ is a source and $u_{3,4n}$ is a sink, $Q$ contains a $u$-$w$ path.
	Further, as the paths in $H_2$ are directed from $v_1$ to $v_2$, $Q$ must visit $v_1$ before $v_2$.
	Let $P$ be the subpath of $Q$ from $u$ to $w$.
	By replacing the subpath of $P$ between $v_1$ and $v_2$ with $v$, we obtain a $u$-$w$ path in $D$ that visits $v$.

	Otherwise, assume towards a contradiction that $Q$ avoids one of the digraphs $H_1$, $H_2$ or $H_3$.
	Let $\bar{a}$ be a minimum identifying sequence for $Q$.
	At most $n$ of the arcs of $\bar{a}$ are in $D$, since $D$ only has $n$ vertices.
	Further, $\Abs{\A{Q} \cap \A{H_i}} \leq 2n$ as each $H_i$ is acyclic.
	As $Q$ avoids one of these digraphs by assumption, we conclude that $Q$ and hence $\bar{a}$ contains at most $5n$ arcs, a contradiction to the initial assumption that $\Anon[D](Q) \geq 6n$.
\end{proof}

\begin{figure}[t]
	\centering
	\begin{tikzpicture}[xscale=0.75]
		\node[circle, label = below:{{$u_{1,3}$}}, line width = 1.2, fill = black, draw = black]
	(v1) at (3.45, 0){};
\node[circle, label = below:{{$u_{1,5}$}}, line width = 1.2, fill = black, draw = black]
	(v2) at (5.75, 0){};
\node[circle, label = below:{{$u_{1,7}$}}, line width = 1.2, fill = black, draw = black]
	(v3) at (8.05, 0){};
\node[circle, label = below:{{$u_{1,9}$}}, line width = 1.2, fill = black, draw = black]
	(v4) at (10.35, 0){};
\node[circle, label = below:{{$u_{1,11}$}}, line width = 1.2, fill = black, draw = black]
	(v5) at (12.65, 0){};
\node[circle, label = below:{{$u_{1,12}$}}, line width = 1.2, fill = black, draw = black]
	(v6) at (13.8, 0){};
\node[circle, label = below:{{$u_{1,1}$}}, line width = 1.2, fill = black, draw = black]
	(v7) at (1.15, 0){};
\node[circle, label = below:{{$u_{1,2}$}}, line width = 1.2, fill = black, draw = black]
	(v8) at (2.3, 0){};
\node[circle, label = below:{{$u_{1,4}$}}, line width = 1.2, fill = black, draw = black]
	(v9) at (4.6, 0){};
\node[circle, label = below:{{$u_{1,6}$}}, line width = 1.2, fill = black, draw = black]
	(v10) at (6.9, 0){};
\node[circle, label = below:{{$u_{1,8}$}}, line width = 1.2, fill = black, draw = black]
	(v11) at (9.2, 0){};
\node[circle, label = below:{{$u_{1,10}$}}, line width = 1.2, fill = black, draw = black]
	(v12) at (11.5, 0){};
\node[circle, label = below:{{$u_{1,0}$}}, line width = 1.2, fill = black, draw = black]
	(v13) at (0, 0){};
\node[]
	(v45) at (13.8, 1.15) [align = left]{\color{green}{$Q_1$}};
\path[-latex, line width = 0.6, draw = black]
	(v13) to (v7);
\path[-latex, line width = 0.6, draw = black]
	(v7) to (v8);
\path[-latex, line width = 0.6, draw = black]
	(v8) to (v1);
\path[-latex, line width = 0.6, draw = black]
	(v1) to (v9);
\path[-latex, line width = 0.6, draw = black]
	(v9) to (v2);
\path[-latex, line width = 0.6, draw = black]
	(v2) to (v10);
\path[-latex, line width = 0.6, draw = black]
	(v10) to (v3);
\path[-latex, line width = 0.6, draw = black]
	(v3) to (v11);
\path[-latex, line width = 0.6, draw = black]
	(v11) to (v4);
\path[-latex, line width = 0.6, draw = black]
	(v4) to (v12);
\path[-latex, line width = 0.6, draw = black]
	(v12) to (v5);
\path[-latex, line width = 0.6, draw = black]
	(v5) to (v6);
\path[-latex, line width = 0.6, draw = green]
	(v13) .. controls (0.673, 0.761) and (1.875, 0.539) .. (v8);
\path[-latex, line width = 0.6, draw = green]
	(v8) .. controls (2.949, 0.665) and (4.144, 0.535) .. (v9);
\path[-latex, line width = 0.6, draw = green]
	(v9) .. controls (5.46, 0.628) and (6.15, 0.573) .. (v10);
\path[-latex, line width = 0.6, draw = green]
	(v10) .. controls (7.512, 0.647) and (8.702, 0.605) .. (v11);
\path[-latex, line width = 0.6, draw = green]
	(v11) .. controls (9.613, 0.559) and (11.03, 0.608) .. (v12);
\path[-latex, line width = 0.6, draw = green]
	(v12) .. controls (12.127, 0.662) and (13.314, 0.593) .. (v6);

	\end{tikzpicture}
	\caption{The digraph $H_1$ for $n = 3$, used in the reduction of the proof of \cref{thm:path anonymity is NP-hard}.
	Every $u_{1,0}$-$u_{1,12}$ path in $H_1$ contains $2n = 6$ arc-disjoint blocking subpaths.}
	\label{fig:np-hard Hi}
\end{figure}
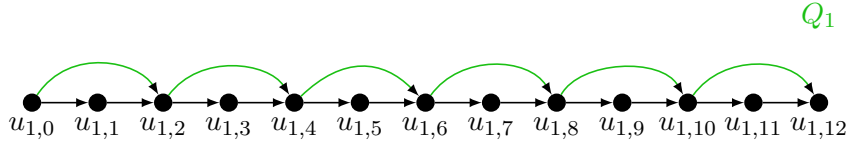

It remains open whether \DPA/ is complete for \NP/,
but in the next section we show that \PathPA/ is $\Sigma_2^p$-complete.

\section{$\Sigma_2^p$-hardness for \PathPA/}
\label{sec:sigma-2-p-hardness-pathpa}

In order to show that \PathPA/ is $\Sigma_2^p$-hard, we define two auxiliary problems and
show that each of them is $\Sigma_2^p$-hard.
Using intermediate problems helps us reduce the complexity of our final reduction.
We provide a reduction from \ShortestImplicantCore/, defined below.

\begin{definition}
	\label{def:implicant}
	Let $\varphi$ be a propositional formula and let $I$ be a set of literals.
	We say that $I$ is an \emph{implicant} for $\varphi$ if $(\bigwedge_{\ell \in I} \ell) \rightarrow \varphi$ is a tautology (that is, it evaluates to true under every assignment of the variables).
\end{definition}

\ShortestImplicantCoreDef

\begin{lemma}[{\cite[Theorem 3]{Umans01}}]
	\label{state:shortest implicant core is sigma_2^p-complete}
	\ShortestImplicantCore/ is $\Sigma_2^p$-complete.
\end{lemma}

The first auxiliary problem is about linkages in a digraph.
We want to decide if there is a subset of the terminals which cannot be linked, whereas we are forced to always connect some fixed pairs.
The last restriction is useful when constructing our gadgets, since
it allows us to adapt the reduction used to show \NP/-hardness for \kLinkage/ \cite{Fortune:1978}, reusing one of their gadgets.

\RestrictedSubsetLinkageDef

\begin{figure}
	\centering
	\begin{tikzpicture}[xscale=1.1]
		\node[circle, label = left:{\color{black}{$v_g^s$}}, line width = 1.2, fill = white, draw = black]
	(v1) at (0, 4.6){};
\node[circle, line width = 1.2, fill = black, draw = black]
	(v2) at (1.15, 4.6){};
\node[circle, line width = 1.2, fill = black, draw = black]
	(v3) at (2.3, 4.6){};
\node[circle, line width = 1.2, fill = black, draw = black]
	(v4) at (2.3, 3.45){};
\node[circle, line width = 1.2, fill = black, draw = black]
	(v5) at (3.45, 3.45){};
\node[circle, line width = 1.2, fill = black, draw = black]
	(v6) at (4.6, 3.45){};
\node[circle, line width = 1.2, fill = black, draw = black]
	(v7) at (5.75, 3.45){};
\node[circle, label = right:{\color{black}{$v_a^s$}}, line width = 1.2, fill = white, draw = black]
	(v8) at (6.9, 2.3){};
\node[circle, line width = 1.2, fill = black, draw = black]
	(v9) at (5.75, 1.15){};
\node[circle, line width = 1.2, fill = black, draw = black]
	(v10) at (4.6, 1.15){};
\node[circle, line width = 1.2, fill = black, draw = black]
	(v11) at (3.45, 1.15){};
\node[circle, line width = 1.2, fill = black, draw = black]
	(v12) at (2.3, 1.15){};
\node[circle, line width = 1.2, fill = black, draw = black]
	(v13) at (2.3, 0){};
\node[circle, line width = 1.2, fill = black, draw = black]
	(v14) at (1.15, 0){};
\node[circle, label = left:{\color{black}{$v_e^s$}}, line width = 1.2, fill = white, draw = black]
	(v15) at (0, 0){};
\node[circle, label = right:{\color{black}{$v_f^s$}}, line width = 1.2, fill = white, draw = black]
	(v16) at (3.45, 0){};
\node[circle, label = right:{\color{black}{$v_h^s$}}, line width = 1.2, fill = white, draw = black]
	(v17) at (3.45, 4.6){};
\node[circle, label = left:{\color{black}{$v_b^s$}}, line width = 1.2, fill = white, draw = black]
	(v18) at (3.45, 2.3){};
\node[circle, label = left:{\color{black}{$v_c^s$}}, line width = 1.2, fill = white, draw = black]
	(v19) at (1.15, 2.3){};
\node[circle, label = left:{\color{black}{$v_d^{s}$}}, line width = 1.2, fill = white, draw = black]
	(v20) at (0, 2.3){};
\path[-latex, line width = 0.6, draw = black]
	(v1) to node[above] {\color{black}{$a_8^s$}} (v2);
\path[-latex, line width = 0.6, fill = black, draw = black]
	(v2) to node[above] {\color{black}{$a_9^s$}} (v3);
\path[-latex, line width = 0.6, fill = black, draw = black]
	(v3) to node[above] {\color{black}{$a_{10}^s$}} (v4);
\path[-latex, line width = 0.6, draw = black]
	(v4) to node[above] {\color{black}{$a_4^s$}} (v5);
\path[-latex, line width = 0.6, draw = black]
	(v5) to node[above] {\color{black}{$a_3^s$}} (v6);
\path[-latex, line width = 0.6, draw = black]
	(v6) to node[above] {\color{black}{$a_2^s$}} (v7);
\path[-latex, line width = 0.6, draw = black]
	(v7) to node[above] {\color{black}{$a_1^s$}} (v8);
\path[-latex, line width = 0.6, draw = black]
	(v9) to node[below] {\color{black}{$b_1^s$}} (v8);
\path[-latex, line width = 0.6, draw = black]
	(v10) to node[below] {\color{black}{$b_2^s$}} (v9);
\path[-latex, line width = 0.6, draw = black]
	(v11) to node[below] {\color{black}{$b_3^s$}} (v10);
\path[-latex, line width = 0.6, draw = black]
	(v12) to node[below] {\color{black}{$b_4^s$}} (v11);
\path[-latex, line width = 0.6, draw = black]
	(v19) to node[above] {\color{black}{$b_5^s$}} (v12);
\path[-latex, line width = 0.6, draw = black]
	(v19) to node[below] {\color{black}{$a_5^s$}} (v4);
\path[-latex, line width = 0.6, draw = black]
	(v15) to node[below] {\color{black}{$b_8^s$}} (v14);
\path[-latex, line width = 0.6, draw = black]
	(v14) to node[below] {\color{black}{$b_9^s$}} (v13);
\path[-latex, line width = 0.6, draw = black]
	(v13) to node[below] {\color{black}{$b_{10}^s$}} (v12);
\path[-latex, line width = 0.6, draw = black]
	(v11) to node[above] {\color{black}{$b_{11}^s$}} (v16);
\path[-latex, line width = 0.6, draw = black]
	(v13) to node[below] {\color{black}{$b_{12}^s$}} (v20);
\path[-latex, line width = 0.6, draw = black]
	(v3) to node[above] {\color{black}{$a_{12}^s$}} (v20);
\path[-latex, line width = 0.6, draw = black]
	(v18) to node[above] {\color{black}{$b_6^s$}} (v10);
\path[-latex, line width = 0.6, draw = black]
	(v18) to node[above] {\color{black}{$a_6^s$}} (v6);
\path[-latex, line width = 0.6, draw = black]
	(v5) to node[below] {\color{black}{$a_{11}^s$}} (v17);
\path[-latex, line width = 0.6, draw = black]
	(v7) .. controls (4.911, 0.42) and (1.115, 3.142) .. (v14);
\path[-latex, line width = 0.6, draw = black]
	(v9) .. controls (5.195, 3.951) and (1.13, 1.695) .. (v2);

	\end{tikzpicture}
	\caption{A switch gadget $S_s$.}
	\label{figure:linkage-gadget-switch}
\end{figure}
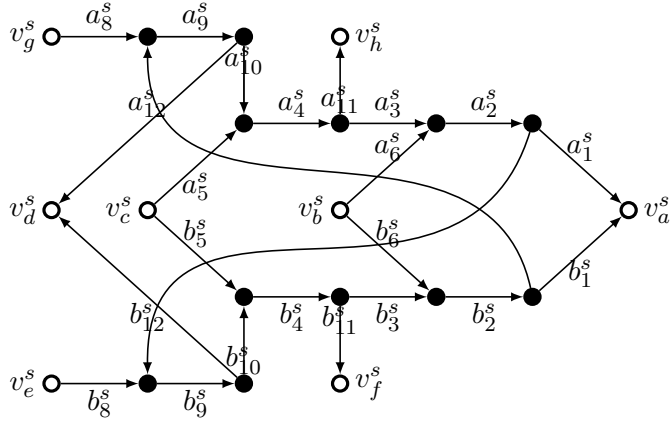

\begin{lemma}[{\cite[Lemma 2]{Fortune:1978}}]
	\label{state:switch-C-A-B-D-linkage}
	Consider the subgraph in \cref{figure:linkage-gadget-switch}.
	Suppose there are two vertex-disjoint paths
	passing through the subgraph, one leaving at vertex $v_{a}^s$ and the other entering at $v_b^s$. 
	Then the path leaving at $v_a^s$ must have entered at $v_c^s$ and
	the path entering at $v_b^s$ must leave at $v_d^s$.
	Further, there is exactly one additional path through the subgraph and it is either
	$a_8^s, a_9^s, a_{10}^s, a_4^s, a_{11}^s$
	or
	$b_{8}^s, b_9^s, b_{10}^s, b_4^s, b_{11}^s$,
	depending on the actual routing of the path leaving at $v_a^s$.
\end{lemma}
\begin{lemma}
	\label{state:restricted-subset-linkage-sigma_2^p-hard}
	\RestrictedSubsetLinkage/ is $\Sigma_2^p$-hard.
\end{lemma}

\begin{proof}
	We provide a reduction from \ShortestImplicantCore/.
	Let $(\varphi, I, k)$ be a \ShortestImplicantCore/ instance.
	Let $C = \{c_1, \ldots, c_m\}$ be the set of clauses in $\varphi$ and let $x_1, \ldots, x_n$ be the variables in $\varphi$.
	Construct a \RestrictedSubsetLinkage/ instance $(D, T_0, T, k')$ as follows.

	For each variable $x_i$ in $\varphi$, add the vertices $z_i, \bar{x}_i, x_i, y_i$ to $D$.
	If $x_i$ has a positive literal in clause $c_j$, add the vertex $x_{i,j}$, and if it has a negative literal, add the vertex $\bar{x}_{i,j}$.
	Add two $z_i$-$y_i$ paths, one connecting all positive occurrences of $x_i$ (including the vertex $x_i$) and another connecting all of its negative occurrences (including the vertex $\bar{x}_i$).
	If $i < n$, add the arc $(y_i, z_{i+1})$.

	Add a vertex $c_0$ to $D$.
	For each clause $c_j$ add a vertex $c_j$
	For each variable $x_i$ in $c_j$, add a switch $S_{i,j}$ to $D$, together with the arcs $(v_f^{i,j}, c_j)$, $( c_{j - 1}, v_e^{i,j})$.
	If $x_i$ is positive in $c_j$, add the arcs $(x_{i,j}, v_g^{i,j})$ and $(x_{i,j'}, v_h^{i,j})$, where $x_{i,j'}$ is the successor of $x_{i,j}$ in the corresponding $x_i$-$y_i$ path.

	Analogously, if $x_i$ is negative in $c_j$, add the arcs $(\bar{x}_{i,j}, v_g^{i,j})$ and $(\bar{x}_{i,j'}, v_h^{i,j})$, where $\bar{x}_{i,j'}$ is the successor of $\bar{x}_{i,j}$ in the corresponding $x_i$-$y_i$ path.
	If $j < m$, then add the arcs $(v_{a}^{i,j}, v_{c}^{i',j + 1})$ and $(v_{d}^{i,j}, v_{b}^{i',j + 1})$ for each variable $x_{i'}$ in $c_{j + 1}$ and each variable $x_i$ in $c_j$.

	Add the vertices $s_1, t_1, s_2, t_2$.
	Add the arcs $(c_m, t_1)$ and $(y_n, c_0)$.
	Add the arcs $(s_1, v_b^{i',1})$, where $i'$ is the lowest index such that the variable $x_{i'}$ is in $c_1$.
	Add the arc $(v_a^{i',1}, t_2)$, where $i'$ is the smallest index such that $x_{i'}$ is in $c_1$. 
	Add the arcs $(s_2, v_c^{i',m})$ and $(v_d^{i',m}, z_1)$, where $i'$ is the largest index such that the variable $x_{i'}$ is in $c_m$.

	For each literal $y_\ell \in I$, add the vertex $y_\ell$ to $D$.
	If $y_\ell$ is a positive occurrence of some variable $x_i$,
	then add the arc $(y_\ell, x_{i})$ to $D$ and the pair $(y_\ell, x_{i})$ to $T$.
	Otherwise, if it is a negative occurrence of $x_i$,
	add the arc $(y_\ell, \bar{x}_i)$ to $D$ and
	the pair $(y_\ell, \bar{x}_i)$ to $T$.

	Finally, add the pairs $(s_1, t_1)$ and $(s_2, t_2)$ to $T_0$ and set $k' = k$.
	This concludes the construction.

	Now assume that the \ShortestImplicantCore/ instance $(\varphi, I, k)$ is a yes-instance and
	let $I' \subseteq I$ be a solution of size at most $k$.

	We construct $T'$ by taking the vertices $y_\ell$ corresponding to the literals in $I'$.
	Assume towards a contradiction that there is a linkage $L$ connecting the terminal pairs in $T' \cup T_0$.
	We show that there is a truth-assigment of the variables in $\varphi$
	in which all literals in $I'$ evaluate to true but
	which does not satisfy $\varphi$.

	Let $P_1$ be the $s_1$-$t_1$ path in $L$ and
	let $P_2$ be the $s_2$-$t_2$.
	By construction, $P_1$ enters the first switch at some $v_b^{i',1}$ and $P_2$ leaves the first switch at some $v_a^{i'',1}$.
	By \cref{state:switch-C-A-B-D-linkage}, $P_1$ and $P_2$ must intersect all switches in order, entering at $v_b^s$ and $v_c^s$ and leaving at $v_d^s$ and $v_a^s$.

	Further, $P_1$ must cross all variable gadgets, as the last switch $S_{s}$ contains exactly one arc leaving $v_d^{s}$, namely $(v_d^s, z_1)$.
	For each variable $x_i$, the path $P_1$ must visit the vertex $x_i$ or $\bar{x}_i$.
	If $P_1$ visits $x_i$, we set $x_i$ to false.
	Otherwise, we set $x_i$ to true.

	In order to connect the terminals in $T'$,
	the linkage $L$ must contain the unique outgoing arc of each $y_\ell \in I'$.
	If $y_\ell$ is a positive occurrence of the variable $x_i$, 
	then $L$ contains the vertex $x_i$ (and $P_1$ must visit $\bar{x}_i$), 
	otherwise it contains the vertex $\bar{x}_i$ (and $P_1$ must visit $x_i$).
	This implies that, in the assignment constructed above, all literals in $I'$ evaluate to true.

	Once $P_1$ reaches $c_0$, it must proceed to $c_m$ by taking, for every clause $c_j$, some $v_e^{i,j}$-$v_f^{i,j}$ path of some switch $S_{i,j}$.
	By \cref{state:switch-C-A-B-D-linkage}, this is only possible if $P_1$ did not take the $v_g^{i,j}$-$v_h^{i,j}$ path when visiting the variable gadget of $x_i$.
	By construction of the switch gadgets, this implies that the variable $x_i$ as set in such a way that the clause $c_j$ is not satisfied by the assignment constructed earlier.

	Since $P_1$ must cross all clause gadgets, it means that there is an assignment in which none of the clauses in $\varphi$ are satisfied.
	This contradicts the original assumption that $I'$ is a implicant for $\varphi$.
	Hence, there is no linkage $L$ connecting the terminal pairs in $T' \cup T_0$.

	For the other direction, let $T' \subseteq T$ be a set of size at most $k$ such that no linkage connecting the terminal pairs of $T' \cup T_0$ exists in $D$.
	For each pair in $T'$, take the corresponding literal from $I$ and add it to $I'$.
	We show that $I'$ is an implicant for $\varphi$.

	Assume towards a contradiction that
	there is some truth assignment
	in which all literals in $I'$ evaluate to true but
	$\varphi$ evaluates to false.
	This implies that
	for every clause $c_j$
	there is some variable $x_i$
	whose corresponding literal in $c_j$ evaluates to false,
	causing $c_j$ to also evaluate to false.

	We construct a linkage $L'$ connecting the terminals in $T'$ by simply taking the unique arc connecting the terminal pairs.
	Note that $I'$ cannot contain a literal $y_\ell$ and its negation at the same,
	as no truth assignment can satisfy both literals at the same time.
	Hence, for each variable $x_i$, there is some $z_i$-$y_i$ which does not contain any terminals of $T'$.

	We construct an $s_1$-$t_1$ path $P_1$ by first crossing all the switches, entering at $v_b^s$ and leaving at $v_d^s$.
	Then, when traversing the gadget corresponding to the variable $x_i$, 
	we choose the $z_i$-$y_i$ path containing $x_i$
	if $x_i$ is set to false and 
	containing $\bar{x_i}$ if $x_i$ is set to true.
	Finally, when $P_1$ reaches $c_0$, we proceed to $c_m$ by taking, for each clause $c_j$, the $v_e^{i,j}$-$v_f^{i,j}$ path of the gadget $S_{i,j}$
	such that $x_i$ is a variable whose assignment cause $c_j$ to evaluate to false.
	Since we have an assignment which does not satisfy $\varphi$,
	such a variable $x_i$ exists.

	The $s_2$-$t_2$ path $P_2$ is constructed by crossing all the switches from $v_c^s$ to $v_a^s$.
	By \cref{state:switch-C-A-B-D-linkage}, the paths $P_1$ and $P_2$ are disjoint.
	Hence, we obtained a linkage $L' \cup \{P_1, P_2\}$ connecting all terminal pairs of
	$T' \cup T_0$, a contradiction to the original assumption that $T'$ is a solution of size $k$.
	Thus, $I'$ is an implicant for $\varphi$, as desired.
\end{proof}

For the hardness reduction for \PathPA/, it is convenient to consider a variant of \RestrictedSubsetLinkage/
in which $T_0 = \emptyset$, because then we do not need to differentiate between $T'$ and $T_0$.

\SubsetLinkageDef

\begin{lemma}
	\label{state:subset-linkage-sigma_2^p-hard}
	\SubsetLinkage/ is $\Sigma_2^p$-hard.
\end{lemma}

\begin{proof}
	We provide a reduction from \RestrictedSubsetLinkage/, which is $\Sigma_2^p$-hard due to \cref{state:restricted-subset-linkage-sigma_2^p-hard}.
	Let $(D, T_0, T, k)$ be a \RestrictedSubsetLinkage/ instance.
	We construct a \SubsetLinkage/ instance $(D_1, T_1, k_1)$ as follows.
	Let $t = \Abs{T}$.

	Construct $D_1$ by starting with a copy of $D$.
	For each $(s_i, t_i) \in T_0 \cup T$, add the vertices $s_i^a$ and  $s_i^a$.
	For each $(s_i, t_i) \in T_0$ and each $(p_j, q_j) \in T$, add the path $s_i^a, w_i^1, w_i^2, \ldots, w_i^t, s_i$, the path $p_j^a, w_i^j, q_j^a$ and the arcs $(p_j^a, p_j), (q_j, q_j^a)$ and $(t_i, t_i^a)$.
	Finally, set $T_1 = \{(s_i^a, t_i^a) \mid (s_i, t_i) \in T_0 \cup T\}$ and
	$k_1 = k + \Abs{T_0}$.
	This completes the construction.

	We first show that, if $(D, T_0, T, k)$ is a yes-instance, the so is $(D_1, T_1, k_1)$.
	Let $T' \subseteq T$ be a solution of size at most $k$.
	We set $T_1' = \{(s_i^a, t_i^a) \mid (s_i, t_i) \in T' \cup T_0\}$.
	Clearly $\Abs{T_1'} \leq k_1$.

	Assume towards a contradiction that there is a linkage $L$ connecting the terminals of $T_1'$ in $D_1$.
	By construction of $D_1$, for each $(s_i, t_i) \in T_0$ there is exactly one
	$s_i^a$-$s_i$ path in $D_1$, namely the path $s_i^a, w_i^1, w_i^2, \ldots, w_i^t, s_i$.
	Since the $s_i^a$-$t_i^a$ path in $L$ must go through $s_i$,
	it must also visit the path above.
	This implies that no path connecting $(s_j^a, t_j^a) \in T_1'$ can use any arc $(s_j^a, w_i^a)$.
	Hence, for each $(s_j, t_j) \in T'$, the $(s_j^a, t_j^a)$ path in $L$ must contain
	a $s_j$-$t_j$ path which is also in $D$.
	Further, $L$ also contains a $s_i$-$t_i$ path for each $(s_i, t_i) \in T_0$.
	This implies that there is a linkage $L'$ connecting the terminal pairs of $T' \cup T_0$,
	a contradiction to the assumption that $T'$ is a solution to the \RestrictedSubsetLinkage/ instance.

	For the other direction, let $T_1'$ be
	a solution of size at most $k_1$ for $(D_1, T_1, k_1)$.
	Let $T_0^\star = \{(s_i^a, t_i^a) \in T_1' \mid (s_i, t_i) \in T_0\}$, 
	$T^\star = \{(s_i^a, t_i^a) \in T_1' \mid (s_i, t_i) \in T\}$ and
	let $T' \subseteq \{(s_i, t_i) \in T \mid (s_i^a, t_i^a) \in T^\star\}$
	be a maximal subset of size at most $k$.
	We show that $T'$ is a solution for $(D, T_0, T, k)$.

	If there is no linkage connecting the terminals of $T_0^\star$ in $D_1$,
	then clearly there is also no linkage connecting the terminals of $T_0$ in $D$, and
	so any subset of $T$ of size at most $k$ is a solution.
	
	If $\Abs{T^\star} \leq k$, then clearly $T'$ is a solution,
	as any linkage connecting $T' \cup T_0$ in $D$ could
	easily be extended to a linkage connecting the terminals of $T_1'$ in $D_1$.

	Assume towards a contradiction that $\Abs{T^\star} > k$.
	In this case, $\Abs{T_0^\star} < \Abs{T_0}$.
	Let $(s_i^a, t_i^a) \in T \setminus T_0^\star$ be such that
	$(s_i, t_i) \in T_0$.
	We construct a linkage $L$ connecting the terminals of $T_1'$ as follows.

	For each $(p_j^a, q_j^a) \in T^\star$,
	add the path $p_j^a, w_i^j, q_j^a$ to $L$.
	Now add a linkage $L_0$ connecting the terminals of $T_0^\star$ in $D_1$.
	As discussed above, this linkage must exist.
	Furthermore, $L_0$ does not intersect any of the paths previously added to $L$,
	since $(s_i^a, t_i^a) \not \in T_0^\star$.
	Hence, $L$ connects all the terminals of $T_1'$, a contradiction to the assumption that
	$T_1'$ is a solution.
\end{proof}

\begin{figure}
	\centering
	\vspace{-0.25cm}
	\begin{tikzpicture}
		\node[circle, label = above:{{$x^{i,c,d}_1$}}, line width = 0.6, fill = black]
	(n0) at (3.45, 1.15){};
\node[circle, label = above:{{$x^{i,c,d}_0$}}, line width = 0.6, fill = black]
	(n11) at (2.3, 1.15){};
\node[circle, label = above:{{$x^{i,c,d}_2$}}, line width = 0.6, fill = black]
	(n12) at (4.6, 1.15){};
\node[circle, label = above:{{$x^{i,c,d}_3$}}, line width = 0.6, fill = black]
	(n13) at (5.75, 1.15){};
\node[circle, label = above:{{$x^{i,c,d}_9$}}, line width = 0.6, fill = gray]
	(n16) at (6.9, 2.3){};
\node[circle, label = below:{{$x^{i,c,d}_4$}}, line width = 0.6, fill = black]
	(n17) at (8.05, 1.15){};
\node[circle, label = below:{{$x^{i,c,d}_5$}}, line width = 0.6, fill = black]
	(n18) at (9.2, 1.15){};
\node[circle, label = below:{{$x^{i,c,d}_6$}}, line width = 0.6, fill = black]
	(n19) at (10.35, 1.15){};
\node[circle, label = above:{{$x^{i,c,d}_7$}}, line width = 0.6, fill = black]
	(n20) at (11.5, 1.15){};
\node[circle, label = above:{{$x^{i,c,d}_8$}}, line width = 0.6, fill = black]
	(n1) at (12.65, 1.15){};
\node[circle, label = below:{{$y^{i,c,d}_{0}$}}, line width = 0.6, fill = black]
	(n2) at (12.65, 0){};
\node[circle, label = below:{{$y^{i,c,d}_{1}$}}, line width = 0.6, fill = black]
	(n3) at (11.5, 0){};
\node[circle, label = below:{{$y^{i,c,d}_{2}$}}, line width = 0.6, fill = black]
	(n4) at (10.35, 0){};
\node[circle, label = below:{{$y^{i,c,d}_{3}$}}, line width = 0.6, fill = black]
	(n5) at (9.2, 0){};
\node[circle, label = below:{{$y^{i,c,d}_{4}$}}, line width = 0.6, fill = black]
	(n6) at (8.05, 0){};
\node[circle, label = below:{{$y^{i,c,d}_{5}$}}, line width = 0.6, fill = black]
	(n7) at (6.9, 0){};
\node[circle, label = below:{{$z^{i,c,d}_{0}$}}, line width = 0.6, fill = black]
	(n8) at (2.3, 0){};
\node[circle, label = below:{{$z^{i,c,d}_{1}$}}, line width = 0.6, fill = black]
	(n9) at (1.15, 0){};
\node[circle, label = below:{{$z^{i,c,d}_{2}$}}, line width = 0.6, fill = black]
	(n10) at (0, 0){};
\node[circle, label = above:{{$z^{i,c,d}_{3}$}}, line width = 0.6, fill = black]
	(n14) at (0, 1.15){};
\node[circle, label = above:{{$z^{i,c,d}_{4}$}}, line width = 0.6, fill = black]
	(n15) at (1.15, 1.15){};
\path[-latex, line width = 1.201, draw = black]
	(n0) to (n12);
\path[-latex, line width = 1.201, draw = black]
	(n12) to (n13);
\path[-latex, line width = 1.201, draw = black]
	(n17) to (n18);
\path[-latex, line width = 1.201, draw = black]
	(n19) to (n20);
\path[-latex, line width = 1.201, draw = black]
	(n20) to (n1);
\path[-latex, line width = 1.201, draw = black]
	(n18) to (n19);
\path[-latex, line width = 1.201, draw = yellow]
	(n13) to (n17);
\path[-latex, line width = 1.201, draw = cyan]
	(n7) to (n12);
\path[-latex, line width = 1.201, draw = cyan]
	(n18) .. controls (7.457, 2.885) and (6.214, 2.886) .. (n12);
\path[-latex, line width = 1.201, draw = cyan]
	(n19) .. controls (8.85, 2.637) and (8.197, 2.115) .. (n13);
\path[-latex, line width = 1.201, draw = cyan]
	(n20) .. controls (11.461, -0.158) and (4.528, 0.148) .. (n0);
\path[-latex, line width = 1.201, draw = black]
	(n11) to (n0);
\path[-latex, line width = 1.201, draw = black]
	(n1) to (n2);
\path[-latex, line width = 1.201, draw = black]
	(n2) to (n3);
\path[-latex, line width = 1.201, draw = black]
	(n3) to (n4);
\path[-latex, line width = 1.201, draw = black]
	(n4) to (n5);
\path[-latex, line width = 1.201, draw = black]
	(n5) to (n6);
\path[-latex, line width = 1.201, draw = black]
	(n9) to (n10);
\path[-latex, line width = 1.201, draw = gray]
	(n13) to (n16);
\path[-latex, line width = 1.201, draw = gray]
	(n16) to (n17);
\path[-latex, line width = 1.201, draw = black]
	(n10) to (n14);
\path[-latex, line width = 1.201, draw = black]
	(n14) to (n15);
\path[-latex, line width = 1.201, draw = black]
	(n8) to (n9);
\path[-latex, line width = 1.201, draw = black]
	(n15) to (n11);
\path[-latex, line width = 1.201, draw = black]
	(n6) to (n7);

	\end{tikzpicture}
			\caption{A gadget $X^{i,c,d}$.
	The arc $(x^{i,c,d}_3, x^{i,c,d}_4)$ is a blocking subpath for the path $P$
	constructed in the proof of \cref{state:path-pa-sigma-2-p-hard}
	because of the bypass $x^{i,c,d}_3, x^{i,c,d}_9, x^{i,c,d}_4$.}
	\label{figure:sigma-2-p-hardness-path-pa-gadget-X}
\end{figure}

We can now provide our main reduction.
One of the biggest challenges in the construction is to use a single conflicting ear $Q$
to count how many pairs from $T'$ were already taken.
Further, it is not clear how to model taking an arc into an ear-identifying sequence $\bar{a}$
as a choice of some terminal $(s_i,t_i)$, and we instead take terminal pairs based on 
arcs which are not taken into $\bar{a}$.

\begin{theorem}
	\label{state:path-pa-sigma-2-p-hard}
	\PathPA/ is $\Sigma_2^p$-hard.
\end{theorem}
\begin{proof}
	We provide a reduction from \SubsetLinkage/, which is $\Sigma_2^p$-hard by \cref{state:subset-linkage-sigma_2^p-hard}.
	Let $(D, T, k)$ be a \SubsetLinkage/ instance.
	We construct an instance $(D_1, P, k_1)$ of $\PathPA/$ as follows.
	Let $t = \Abs{T}$,
	let $d_{\text{max}} = 2k+1$ and
	let $c_{\text{max}} = k+2$.

	Sort \(T\) arbitrarily as \(\{(s_1, t_1), (s_2, t_2), \ldots, (s_t, t_t)\}\) and
	add a copy of $D$ to $D_1$.

	For each $(s_i, t_i) \in T$ and
	each $c \in \Set{0, \ldots, c_{\text{max}}}$
	construct a gadget $G^{i,c}$ and
	a subpath \(P^{i,c}\) of \(P\) as follows.

	For each $0 \leq d \leq d_{\text{max}}$
	let \(P^{i,c,d}\) be the path
	\(z^{i,c,d}_0, \ldots, z^{i,c,d}_5, x^{i,c,d}_0, \ldots, x^{i,c,d}_8, y^{i,c,d}_0, \ldots y^{i,c,d}_6\) and
	let $X^{i,c,d}$ be the digraph consisting of
	\(P^{i,c,d}\),
	the path
	\(x^{i,c,d}_3, b^{i,c,d}_2, x^{i,c,d}_4\) and
	the arc
	\((y^{i,c,d,}_6, b_1^{i,c,d})\).

	If \((i,c) \neq (t, c_{\text{max}})\), \(d < d_{\text{max}}\) and \(c \geq 2\),
	add the arc \((b_1^{i,c,d}, z_0^{i,c,d+1})\).
	Note that this creates a bypass from \(y_6^{i,c,d}\) to \(z_0^{i,c,d+1}\).

	If $c = 0$, add
	$X^{i,0,0}$ to \(G^{i,0}\)
	together with the vertex \(w_1^i\) and
	the arc \((w_1^i, z_0^{i,0,0})\).
	Set \(P^{i,0}\) to the path \((w_1^i, z_0^{i,0,0}) \cdot P^{i,0,0}\).

	If \(c = 1\), add
	\(X^{i,1,0}\) to \(G^{i,1}\)
	together with the vertices \(a_0^i, a_1^i\) and
	the arcs \((x_2^{i,1,0}, s_i)\), \((y_6^{i,1,0}, a_0^i)\), \((a_0^{i}, a_1^i)\), \((t_i, a_1^i)\) and
	\((w_1^i, z_0^{i,1,0})\).
	Note that \(s_i\) and \(t_i\) are the terminals of the pair \((s_i, t_i) \in T\),
	which lie in the copy of \(D\).
	Define \(P^{i,1}\) as the path
	\(P^{i,1,0}, (y_6^{i,1,0}, a_0^i), (a_0^i, a_1^i)\).

	If \(c \geq 2\),
	add \(X^{i,c,d}\) for each \(0 \leq d \leq d_{\text{max}}\).
	For each \(0 \leq d \leq d_{\text{max}} - 1\),
	add the arc \((y_6^{i,c,d}, z_0^{i,c,d+1})\).

	If \(c < c_{\text{max}}\), add the arc \((y_6^{i,c,d_{\text{max}}}, z_0^{i,c+1,0})\),
	otherwise, if \(i < t\), add the arc \((y_6^{i,c_{\text{max}},d_{\text{max}}}, z_0^{i+1,0,0})\).

	We now add arcs between the \(G^{i,c}\).
	For each \((s_i, t_i) \in T\), we add the following arcs.
	First, add the path \(F^i_{\text{start}}\), which is defined as
	\(w_i^i, y_3^{i,k+1,d_{\text{max}}}, y_3^{i,k+1,d_{\text{max}}-1}, \ldots, y_3^{i, k+1, 0}\).

	For each \(0 \leq d,f \leq d_{\text{max}}\)
	and
	each \(2 \leq c < c_{\text{max}}\),
	add the paths
	\begin{align*}
		A^{i,c,d,f}_{\text{taken}}      & \coloneqq (x^{i,c,d}_{6}, x^{i,c,d}_{3}, x^{i,c,d}_{4}, x^{i,c,d}_{5}, x^{i,c,d}_{2}, z^{i,c+1,f}_{1}, y^{i,c,d}_{6}),
		\\ A^{i,c,d,f}_{\text{free}}    & \coloneqq (x^{i,c,d}_{2}, \ldots, x^{i,c,d}_{7},x^{i,c,d}_{1}, z^{i,c+1,f}_{3}, y^{i,c,d}_{6}),
		\\ C^{i,c,d}_{\text{unknown}}   & \coloneqq (x^{i,c,d}_{8}, y^{i,c,d}_{0}, \ldots, y^{i,c,d}_{5}, x^{i,c,d}_{2}),
		\\ E^{i,c+1,f,d}_{\text{taken}} & \coloneqq (z_0^{i,c+1,f}, z_5^{i,c,d}, x_6^{i,c+1,f} ),
		\\ E^{i,c+1,f,d}_{\text{free}} & \coloneqq (z_2^{i,c+1,f}, z_4^{i,c,d}, x_8^{i,c+1,f} ).
	\end{align*}

	If \(i < t\), then
	for each \(3 \leq c \leq c_{\text{max}} - 1\)
  add the paths
	\begin{align*}
		C^{i,c,d}_{\text{taken}}   & \coloneqq
		(x^{i,c,d}_{8}, y^{i,c,d}_{0}, y^{i,c,d}_{1}, y^{i+1,c,d}_{2}, y^{i+1,c-1,d}_{2}, x^{i,c,d}_{6}),
		\\ C^{i,c,d,f}_{\text{free}} & \coloneqq
		(x^{i,c,d}_{8}, y^{i,c,d}_{0}, y^{i,c,d}_{1}, y^{i+1,c,f}_{2}, y^{i+1,c,f}_{3}, y^{i+1,c,f}_{4}, x^{i,c,d}_{2}).
	\end{align*}

	For each \(1 \leq i < j \leq t\) and
	each \(3 \leq c \leq c_{\text{max}}\)
	add the paths
	\begin{align*}
		D^{i,j,c}_{\text{choice}}  \coloneqq &
		(y^{i,c,0}_{3}, y^{j,c-1,d_{\text{max}}}_{3}, y^{j,c-1,d_{\text{max}} - 1}_{3}, y^{j,c-1,0}_3),
		\\ D^{j,i}_\text{back}  \coloneqq &
		(y^{j,2,0}_{3}
		, y^{i,k,d_{\text{max}}}_{3}
		, y^{i,k,d_{\text{max}} - 1}_{3}
		, \ldots
		, y^{i,k,0}_{3}
		, x^{i,0,0}_{6}
		)
		\\ & \cdot
		A_{\text{taken}}^{i,0,0,0}
		\cdot
		E_{\text{taken}}^{i,1,0,0}
		\\ & \cdot
		( x^{i,1,0}_{6}
		, x^{i,1,0}_{3}
		, x^{i,1,0}_{4}
		, x^{i,1,0}_{5}
		, x^{i,1,0}_{2}
		, s_i ).
	\end{align*}

	Add the vertices \(a_0^0, b^0\) and
	the arcs 
	\((a_0^0, z_0^{1,0,0})
	, (a_0^0, b^0)
	, (b^0, z_0^{1,0,0})\)
	to \(D_1\).
	Let \(P^0\) be the path \(a_0^0, z_0^{1,0,0}\).

	For each \(1 \leq i \leq t\)
	let \(P^i\) be the path obtained by concatenating
	\(P^{i,0}, P^{i,1}, \ldots, P^{i,c_{\text{max}}}\).
	Finally, let $P \coloneqq P^0 \cdot P^1
	\cdot (\End{P^1}, \Start{P^2})
	\cdot P^2
	\cdot (\End{P^2}, \Start{P^3})
	\cdot \ldots
	\cdot P^t$.
	We set $k_1 = t(k+2)(2k + 1) + t - k + 1$.
	This completes the construction of the reduced instance $(D_1, P, k_1)$.

	We now show that the reduction is correct.
	That is, we show that \((D,k)\) is a yes-instance of \SubsetLinkage/ if, and only if,
	\((D_1, P, k_1)\) is a yes-instance of \PathPA/.
	Let \(B' = \Set{(a_0^0, a_1^0)}
	\cup \Set{(x_3^{i,c,d}, x_4^{i,c,d}) \mid 1 \leq i \leq t, 0 \leq c \leq c_{\text{max}}, 0 \leq d \leq d_{\text{max}}}
	\cup \Set{(y_6^{i,c,d}, z_0^{i,c,d+1}) \mid 1 \leq i \leq t, 0 \leq c \leq c_{\text{max}}, 0 \leq d < d_{\text{max}}}\).

	\(\bm{(\Rightarrow):}\)
	Let \(T' \subseteq T\) be a solution to the original instance.
	Without loss of generality, we assume that \(\Abs{T'} = k\).

	We construct an ear-identifying sequence \(\bar{a}\) for \(P\) as follows.
	First, take the arc \(a_0^0, z_0^{1,0,0}\).
	For each \(G^{i,c}\), take the arcs
	\((x_3^{i,c,0}, x_4^{i,c,0}), (x_3^{i,c,1}, x_4^{i,c,1}), \ldots, (x_3^{i,c,d_{\text{max}}}, x_4^{i,c,d_{\text{max}}})\)
	if \(c \geq 0\) and
	the arc \((x_3^{i,c,0}, x_4^{i,c,0})\) otherwise.
	For each \((s_i, t_i) \in T \setminus T'\), take the arc
	\((a_0^i, a_1^i)\).
	The arcs are then sorted according to their order on \(P\).

	Assume towards a contradiction that a conflicting ear \(Q\) for \(P, \bar{a}\) exists.
	We show that \(Q\) contains as a subgraph
	a linkage \(\mathcal{L}\) in \(D\) connecting the terminal pairs of \(T'\).

	Since \(\bar{a}\) contains 
	\(a_0^0, a_1^0\), \(\Start{Q} = \Start{P}\).
	Let \(a'\) be the first arc along \(Q\) which is not in \(P\).
	We argue that \(a' = \Start{F_{\text{start}}^{i,c,d}}\) for some choice of \(i,c,d\).
	Assume that this is not the case.
	Since \(a'\) is the first such arc along \(Q\),
	it cannot jump backwards along \(P\), as then \(Q\) would not be an ear.
	Now consider the different possible choices for \(a'\).

	\begin{CaseDistinction}

		\Case \(a'\) is some \((x_1^{i,c,d}, z_3^{i,c+1,f})\).

		The next arc that \(Q\) must visit after \(a'\) is \((x_3^{i,c,d}, x_4^{i,c,d})\).
		Towards this end, \(Q\) visits \(x_2^{i,c,d}\) before \((x_3^{i,c,d}, x_4^{i,c,d})\).
		It cannot reach \(x_2^{i,c,d}\) through \(x_5^{i,c,d}\),
		as \(Q\) would then have to end on \(x_4^{i,c,d}\) before visiting the other arcs of \(\bar{a}\).
		Hence, \(Q\) must contain \((y_5^{i,c,d}, x_2^{i,c,d})\).
		From there, \(Q\) must go along \(P\) until reaching \(y_2^{i,c,d}\).
		However, now it is not possible for \(Q\) have arrived at \(y_5^{i,c,d}\) before visiting \(y_2^{i,c,d}\).
		Hence, this case cannot happen.

		\Case \(a'\) is some \((z_5^{i,c,d}, x_6^{i,c+1,f})\).

		The next arc that \(Q\) must visit after \(a'\) is \(x_3^{i,c,d}, x_4^{i,c,d}\).
		After \(a'\), \(Q\) cannot take the arc \((x_6^{i,c+1,f}, x_2^{i,c+1,f})\),
		as it would then have to visit \((x_3^{i,c+1,f}, x_4^{i,c+1,f})\)
		before \(x_3^{i,c,d}, x_4^{i,c,d}\).
		Hence, it must visit \(x_7^{i,c+1,f}\) and \(x_8^{i,c+1,f}\) afterwards,
		as it also cannot take \((x_7^{i,c+1,f}, x_1^{i,c+1,f})\) for the same reason.
		However, this implies that \(Q\) must reach
		\((x_3^{i,c+1,f}, x_4^{i,c+1,f})\) by going through
		\(x_2^{i,c+1,f}\) first
		and, hence, \(Q\) must end on \(x_5^{i,c+1,f}\) without visiting the other arcs of \(\bar{a}\), a contradiction.

		\Case \(a'\) is some \((y_1^{i,c,d}, y_2^{i+1,c,f})\).

		The next arc that \(Q\) must visit after \(a'\) is either \((y_6^{i,c,d}, z_0^{i,c,d+1})\) or
		\((y_6^{i,c,d}, z_0^{i,c+1,0})\) (if \(d = d_{\text{max}}\)).
		Without loss of generality, the next arc is \((y_6^{i,c,d}, z_0^{i,c,d+1})\).
		The other case follows analogously.

		After visiting \(y_2^{i+1,c,f}\),
		\(Q\) must go to some \(x_6^{i,c,g}\) where \(g > d\).
		We now consider how \(Q\) reaches the arc \((x_3^{i,c,g}, x_4^{i,c,g})\).
		\(Q\) cannot do so by taken \(A_{\text{taken}}^{i,c,g}\), as that would mean visiting
		\((x_3^{i,c,g}, x_4^{i,c,g})\) before 
		\((x_3^{i,c,d}, x_4^{i,c,d})\).
		Hence, \(Q\) must first visit \(x_2^{i,c,g}\) and then \((x_3^{i,c,g}, x_4^{i,c,g})\).
		However, in this case \(Q\) must end on \(x_5^{i,c,g}\) before visiting other arcs of \(\bar{a}\), a contradiction.

		\Case \(a'\) is some \((x_2^{i,1,0}, s_i)\) or some \((x_2^{i,c,d}, z_1^{i,c+1,f})\).

		If \(a' = (x_2^{i,1,0}, s_i)\), set \(c = 1\) and \(d = 0\).
		The next arc that \(Q\) must visit after \(a'\) is \((x_3^{i,c,d}, x_4^{i,c,d})\).
		Since \(Q\) already visited \(x_2^{i,c,d}\),
		it cannot contain the arc \((x_2^{i,c,d}, x_3^{i,c,d})\).
		Hence, it must contain \((x_6^{i,c,d}, x_3^{i,c,d})\).
		However, this implies that \(Q\) must end on \(x_5^{i,c,d}\),
		as it cannot cross \(x_6^{i,c,d}\) twice.
		Thus \(Q\) cannot visit all arcs of \(\bar{a}\), a contradiction.
	\end{CaseDistinction}
	
	Hence, \(a'\) must be \(\Start{F_{\text{start}}^i}\) for some \(1 \leq i \leq t\).
	Let \(i_1\) be such an \(i\).

	After taking \(F_{\text{start}}^{i_1}\), for each \(2 \leq c \leq k+1\)
	there are \(i,j\) such that
	\(Q\) visits
	\(D_{\text{choice}}^{i,j,c}\).
	Let \(i_{2}, i_{3}, \ldots, i_{k}\)
	be the sequence of the indices \(j\) visited for each \(c\) above.

	\(Q\) must then take \(D_{\text{back}}^{i_{k}, i_1}\).
	After arriving at \(s_{i_1}\), it must take a path in \(D\)
	to \(t_{i_1}\) in order to visit the next arc of \(\bar{a}\).

	Now \(Q\) must take some \(E_{\text{free}}^{i,2,0,d}\) to proceed.
	However, on every subpath \(P^i\), \(Q\) must take some
	\(C_{\text{free}}^{i,c,d,f}\) or some
	\(C_{\text{taken}}^{i,c,d}\) for the smallest value \(c\)
	such that
	\(Q\) took some \(C_{\text{taken}}^{j,c+1,g}\) before.
	This guarantees that, for each \(j\) such that \(i_j\) is defined,
	\(Q\) must take \(A_{taken}^{i_j,c+1,g}\).

	Whenever \(Q\) takes \(A_{\text{taken}}^{i,c,d,f}\),
	it must later take \(E_{\text{taken}}^{i,c+1,f,g}\), 
	since \(Q\) visited \(z_1^{i,c+1,f}\).

	Similarly, by taking some \(E_{\text{taken}}^{i,c,f,d}\),
	\(Q\) must follow with \(A_{\text{taken}}^{i,c,f,g}\)
	as this is the only way to visit \((x_3^{i,c,f}, x_4^{i,c,f})\)
	from \(x_6^{i,c,f}\).

	This process only stops when \(Q\) reaches some \(s_i\), at which point it must take
	an \(s_i\)-\(t_i\) path in \(D\).

	Hence, \(Q\) contains a linkage \(\mathcal{L}\) connecting all terminal pairs of \(T'\),
	a contradiction to the initial assumption that \(T'\) is a solution to the original instance.

	\textbf{\(\bm{(\Leftarrow):}\)}
	Let $\bar{a}$ be an ear-identifying sequence for $P$ of length at most $k_1$.
	We construct a set $T' \subseteq T$ of size $k$ as follows.
	Without loss of generality, we assume that $\bar{a}$ has length exactly $k_1$,
	as otherwise we can add arbitrary arcs of $P$ to $\bar{a}$.
	
	There is a sequence $\mathcal{P}$ of paths $P^i$ such that
	no path in $\mathcal{P}$ contains more than $(2k + 1)(k + 2) + k$ arcs of $\bar{a}$ and
	there are at least $k$ paths of $\mathcal{P}$ which
	hit exactly $(2k + 1)(k + 2)$ arcs of $\bar{a}$.

	Let $\mathcal{P}' \subseteq \mathcal{P}$ be a set of $k$ paths
	which hit exactly $(2k + 1)(k + 2)$ arcs of $\bar{a}$.
	Note that all arcs hit by the paths of $\mathcal{P}'$
	are in $A'$.
	We set $T' = \Set{(s_i, t_i) \mid P^i \in \mathcal{P}'}$.

	Assume towards a contradiction that there is a linkage \(\mathcal{L}\)
	in $D$ connecting the terminal pairs of $T'$.
	We construct a conflicting ear $Q$ for $(P, \bar{a})$ as follows.
	For each \(P_i \in \mathcal{P}\) and each \(0 \leq c \leq c_{\text{max}}\),
	choose some \(d^{i,c}\) such that
	the only vertices of \(Q^{i,c,d^{i,c}}\)
	which are incident to some arc of \(\bar{a}\)
	are \(x_3^{i,c,d^{i,c}}\) and
	\(x_4^{i,c,d^{i,c}}\).
	Because \(d_{\text{max}} \geq 2k + 1\) and
	\(P_i\) hits at most \(k\) arcs of \(\bar{a}\) which are not in \(A'\),
	such a number \(d^{i,c}\) exists.

	Let \(\Set{P^{i_{1}}, P^{i_{2}}, \ldots, P^{i_{k}}} \coloneqq \mathcal{P}'\)
	be the paths of \(\mathcal{P}'\) ordered according to their occurrence along \(P\).
	We partition \(Q\) into subpaths and construct each subpath as follows.
	Let $Q_0$ be the subpath of \(P\)
	from \(\Start{P}\) to $w^{i_1}_1$.

	Construct \(Q_1\) as follows.
	Take \(F_{\text{start}}^{i_1}\).
	Iterate for each \(2 \leq j \leq k\).
	On step \(j\), take \(D_{\text{choice}}^{i_{j-1}, i_j, k - j + 2}\).
	After completing the iteration above, we arrive at \(y_3^{i_{k}, 2, 0}\).
	Now take \(D_{\text{back}}^{i_{k}, i_1}\).
	Take the \(s_{i_1}\)-\(t_{i_1}\) path in \(\mathcal{L}\) and
	then go to \(a_1^{i_1}\) and \(z_0^{i_1,2,0}\).

	Now for each \(i_1 \leq i \leq i_{\text{max}}\) and
	for each \(0 \leq c \leq c_{\text{max}}\), construct
	\(Q^{i,c}\) as follows.
	Start by taking the subpath of \(P\)
	from \(z_0^{i,c,0}\) to
	\(\Start{E_{\text{free}}^{i,c,d^{i,c}}}\).
	For the remainder of \(Q^{i,c}\),
	we distinguish between two cases.

	\begin{CaseDistinction}
		\Case \(P_{i+1} \in \mathcal{P}'\).

		Choose \(c'\) so that \(i_{k - c' + 2} = i + 1\) and
		let \(c'' = k - c' + 2\).
		That is,
		\(Q\) contains the path \(D_{\text{choice}}^{i_{c'' - 1}, i_{c''}, c'}\).

		If \(c' < c\), complete \(Q^{i,c}\) by taking the path
		\(E_{\text{free}}^{i,c,d^{i,c}} \cdot C_{\text{unknown}}^{i,c,d^{i,c}, d^{i+1,c}} \cdot A_{\text{free}}^{i,c,d^{i,c}, d^{i,c+1}}\).

		If \(c' = c\), complete \(Q^{i,c}\) by taking the path
		\(E_{\text{free}}^{i,c,d^{i,c}} \cdot C_{\text{taken}}^{i,c,d^{i,c}, d^{i+1,c}} \cdot A_{\text{taken}}^{i,c,d^{i,c}, d^{i,c+1}}\).

		If \(c' > c\), complete \(Q^{i,c}\) by taking the path
		\(E_{\text{taken}}^{i,c,d^{i,c}} \cdot A_{\text{taken}}^{i,c,d^{i,c}, d^{i,c+1}}\).

		\Case \(P_{i+1} \in \mathcal{P} \setminus \mathcal{P}'\).

		Choose \(c'\) to be the largest index such that \(i_{c'} < i\).
		If \(i = i_1\), then set \(c' \coloneqq k + 1\) instead.

		If \(c' < c\), complete \(Q^{i,c}\) by taking the path
		\(E_{\text{free}}^{i,c,d^{i,c}} \cdot C_{\text{unknown}}^{i,c,d^{i,c}, d^{i+1,c}} \cdot A_{\text{free}}^{i,c,d^{i,c}, d^{i,c+1}}\).

		If \(c' \geq c\), complete \(Q^{i,c}\) by taking the path
		\(E_{\text{free}}^{i,c,d^{i,c}} \cdot C_{\text{free}}^{i,c,d^{i,c}, d^{i+1,c}} \cdot A_{\text{free}}^{i,c,d^{i,c}, d^{i,c+1}}\).

	\end{CaseDistinction}

	Now set \(Q \coloneqq Q_0 \cdot Q_1 \cdot \Pi_{i=1}^t\Pi_{c=0}^{c_{\text{max}}}Q^{i,c}\).
	We argue that \(Q\) visits all arcs of \(\bar{a}\) in order.

	The path \(Q_0\) covers all arcs of \(\bar{a}\)
	which come before \(P_{i_1}\).
	The path \(Q_1\) covers \((x_3^{i_1,0,0}, x_4^{i_1,0,0})\).
	For \(i \geq i_1\), the remaning arcs are covered as follows.

	For each \(Q^{i,c}\),
	we take either some \(A_{\text{free}}^{i,c,d}\)
	or some \(A_{\text{taken}}^{i,c,d}\).
	Hence, we visit the arc \((x_3^{i,c,d}, x_4^{i,c,d})\).
	Further, we visit the other arcs \((x_3^{i,c,f}, x_4^{i,c,f})\),
	where \(f \neq d\), by following along \(P\).
	Additionally, if \((s_i, t_i) \in T \setminus T'\), we visit the arc \((a_0^i, a_1^i)\).

	Since we choose \(d^{i,c}\) such that
	\(G^{i,c,d^{i,c}}\) does not contain arcs of \(\bar{a}\)
	beyond those in \(A'\),
	we visit all arcs of \(\bar{a}\) in order.
	As \(Q\) is clearly distinct from \(P\),
	it is a conflicting ear for \((P, \bar{a})\), a contradiction to the assumption that \(\bar{a}\)
	is an ear-identifying sequence for \(P\).
	Hence, the linkage \(\mathcal{L}\) cannot exist, and
	so \(T'\) is a solution for the original \SubsetLinkage/ instance.

										\end{proof}

\section{Remarks}
\label{sec:ear-anonymity:remarks}

Using \cref{state:ear-anonymity-minor-closed,state:identifying-sequence-implies-hitting-set} and the \emph{directed grid theorem} below, it is possible to draw a connection between directed treewidth and ear anonymity.

\begin{theorem}[\cite{kawarabayashi2015directed}]
	There is a computable function $f : \Naturals \rightarrow \Naturals$ such that
	every digraph $D$ with $\DTreewidth(D) \geq f(k)$ contains a cylindrical grid of order $k$ as a butterfly minor,
	where $\DTreewidth(D)$ is the directed tree-width of $D$.
\end{theorem}

It is easy to verify that a cylindrical grid of order $k$ has ear anonymity at least $2k$.
Take any cycle $C$ on the cylindrical grid which is neither the outermost nor the innermost cycle.
Then, a subpath $Q_{i}$ of $C$ from row $i$ to row $i+1 \operatorname{mod} 2k$ is a blocking subpath of a bypass for $C$.
Since $C$ has at least $2k$ internally disjoint blocking subpaths, by \cref{state:identifying-sequence-implies-hitting-set} we have $k \leq \Anon[D](C) \leq \Anon(D)$.
Hence, we obtain the following inequality.

\begin{observation}
	\label{state:ear-anonymity-and-dtw}
	There is a computable function $f : \N \to \N$ such that
	$\DTreewidth(D) \leq f(\Anon(D))$.
\end{observation}

\Cref{state:ear-anonymity-and-dtw} naturally raises the following question.

\begin{question}
	\label{question:ear-anonymity-and-dtw}
	What is the smallest function $f$ such that $\DTreewidth(D) \leq f(\Anon(D))$ holds for all digraphs $D$?
\end{question}

On the other hand, directed acyclic graphs have directed treewidth zero but can have arbitrarily high ear anonymity.
For example, the digraph $H_1$ used in the reduction in the proof of \cref{thm:path anonymity is NP-hard} (see \cref{fig:np-hard Hi}) is acyclic and $\Anon(H_1) = 2n$.
Thus, there is no function $f : \N \to \N$ for which
$\Anon(D) \leq f(\DTreewidth(D))$ holds
for all digraphs $D$.

Since \DPA/ is in \P/ if the input digraph is acyclic, it is natural to ask what is the parameterized complexity of \DPA/ when parameterized by directed treewidth.

\begin{question}
	\label{question:ear-ear-anonymity-XP-wrt-to-dtw}
	Can \DPA/ be solved in $\Bo(n^{f(\DTreewidth(D))})$ time, where $\DTreewidth(D)$ is the directed treewidth of the input digraph $D$?	
\end{question}

It is still unclear where exactly in the polynomial hierarchy \DPA/ lies.
On the one hand, \PathPA/ looks like a subproblem of \DPA/, yet if the digraph has very high ear anonymity, then there are many ``correct'' guesses for
some ear of high anonymity, which could make the problem easier, and not harder, than \PathPA/.

\begin{question}
	\label{question:ear-anonymity-in-NP}
	Is \PAlong/ in \NP/?
\end{question}

Finally, one could also ask if a phenomenon similar to the directed grid theorem also occurs with ear anonymity.
That is, while a path with high ear anonymity witnesses that an acyclic digraph has high ear anonymity, is there also some witness which gives us an upper bound on the ear anonymity of the same digraph?

\begin{question}
	\label{question:witness-for-low-ear-anonymity}
	Is there some ``small'' witness $W$ and some function $f$ which allow us to efficiently verify that $\Anon(D) \leq f(W)$?
\end{question}

\begin{question}
	Can we solve \kLinkage/ in $f(k)n^{g(\Anon(d))}$-time? (In general?) (On DAGs?)
\end{question}

\bibliographystyle{plainnat}
\bibliography{references}
\end{document}